\documentclass[11pt,letterpaper,reqno]{amsart}
\usepackage[margin=1in]{geometry}
\usepackage{amsmath,amsthm,amssymb,graphicx,paralist}
\usepackage[authoryear,sort&compress]{natbib}
\bibpunct{[}{]}{,}{n}{,}{,}
\usepackage{bm}
\usepackage{tikz}
\usepackage{tikz-cd}
\usepackage{comment}
\usepackage[colorlinks=false,bookmarksdepth=0]{hyperref}
\usepackage[all]{xy}
\entrymodifiers={+!!<0pt,\fontdimen22\textfont2>}
\theoremstyle{plain}
\newtheorem{theorem}{Theorem}[section]
\newtheorem{proposition}[theorem]{Proposition}
\newtheorem{corollary}[theorem]{Corollary}

\theoremstyle{definition}

\theoremstyle{remark}
\newtheorem{remark}[theorem]{Remark}

\numberwithin{equation}{section}

\def\rank{\mathop{\mathrm{rank}}\nolimits}

\def\Cr{\mathop{\mathrm{Cr}}\nolimits}

\newcommand{\Q}{\mathcal{Q}}
\newcommand{\Rr}{\mathcal{R}}
\newcommand{\N}{\mathcal{N}}
\newcommand{\M}{\mathcal{M}}

\newcommand{\Pp}{\mathcal{P}}

\newcommand{\bq}{\mathbf{q}}
\newcommand{\bp}{\mathbf{p}}

\begin{document}

\title[Discrete HJ]{A Morse-Family Integrator for Hamilton--Jacobi Dynamics Across Caustics}




\author{
F. Jim\'enez Alburquerque
\and
M. Leok
\and
C. Sard\'on
\and
X. Zhao
}
\begin{abstract}
We develop a geometric framework for implicit discrete Hamiltonian systems
based on discrete Morse families, Lagrangian relations, and discrete
analogues of Tulczyjew's triple. The main idea is to regard the
Lagrangian submanifold defining the discrete dynamics, rather than an
explicit symplectic evolution map, as the fundamental geometric object.
This viewpoint naturally accommodates implicit, constrained, and
degenerate discrete systems.

Within this framework, we formulate a Type--II discrete
Hamilton--Jacobi theory in terms of the propagation of Lagrangian
submanifolds between consecutive discrete steps. When these
submanifolds are locally represented by exact one-forms
$dW_k$ and $dW_{k+1}$, the resulting equations provide a discrete
Hamilton--Jacobi relation between consecutive generating functions.
More generally, when the projection onto configuration space becomes
singular and a single-valued generating function ceases to exist, we
show that the evolution can be described by the composition of
Type--II discrete dynamics with Morse families. This yields a
generating family for the propagated Lagrangian submanifold without
requiring the dynamics to be represented as a graph.

As an application, we consider the propagation of optical wavefronts
through fold caustics. A Type--II discrete Hamiltonian yields a
symplectic ray integrator, while a Morse family represents the
multivalued wavefront near the caustic. Their composition provides a
discrete propagation rule for the complete Lagrangian manifold across
the singularity. In this way, the same geometric construction
simultaneously provides a discrete Hamiltonian integrator and a
regular representation of multivalued Hamilton--Jacobi solutions.
\end{abstract}

\maketitle


\section{Introduction}

Hamilton--Jacobi theory provides one of the fundamental links between
Hamiltonian dynamics, generating functions, and the geometry of
Lagrangian submanifolds. In its classical geometric formulation, a
solution of the Hamilton--Jacobi equation determines a Lagrangian
submanifold of phase space through the image of a closed, and locally
exact, one-form. This interpretation places Hamilton--Jacobi theory
naturally within symplectic geometry and makes generating functions a
fundamental tool for the geometric description of Hamiltonian
dynamics.

Variational principles provide a closely related geometric description
of mechanical systems. In the continuous setting, Lagrangian and
Hamiltonian dynamics can be formulated in terms of symplectic
structures, Lagrangian submanifolds, and generating families.
Tulczyjew's triple gives a particularly useful framework connecting
these descriptions and allows regular, singular, constrained, and
implicit systems to be treated geometrically.

In parallel, discrete variational mechanics has become an important
tool for constructing geometric numerical integrators. Starting from
a discrete variational principle, one obtains discrete Euler--Lagrange
equations whose associated evolution maps preserve a discrete
symplectic structure and reproduce many of the qualitative properties
of the underlying continuous dynamics. Since the work of Marsden and
West \cite{Marsden_West_2001}, variational integrators and their
geometric properties have been extensively developed. Discrete
Hamiltonian mechanics provides the corresponding Hamiltonian
description, in which Type--I and Type--II generating functions give
rise to discrete Hamilton equations and symplectic evolution maps
\cite{LeokZhang2011},\cite{leok2011variational}.

The description in terms of ordinary generating functions, however,
depends on the regularity of a suitable projection of the underlying
Lagrangian submanifold. This assumption becomes restrictive for
implicit, constrained, and degenerate systems, where the dynamics is
more naturally described by a Lagrangian relation than by the graph of
a symplectic transformation. A related difficulty appears when the
projection of a regular Lagrangian submanifold onto configuration
space develops singularities. In this case the Lagrangian submanifold
itself may remain regular even though its representation by a
single-valued generating function breaks down.

Morse families provide a natural geometric mechanism for overcoming
this difficulty. They generalize ordinary generating functions by
introducing auxiliary variables and recovering the associated
Lagrangian submanifold through criticality with respect to those
variables. In particular, the Maslov--H\"ormander theorem ensures that
Lagrangian submanifolds admit local representations by Morse families
\cite{BeCaSi09,Be11,Ca91,Ca15}. This viewpoint has proved especially
useful in the geometric treatment of implicit Hamiltonian systems and
singular Legendre transformations, where the Lagrangian submanifold,
rather than a particular choice of generating function, is regarded as
the fundamental geometric object \cite{EsenLeonSardon2018}.

The purpose of the present work is to develop this viewpoint in the
discrete setting. We introduce discrete generating families as
discrete analogues of Morse families and combine them with discrete
versions of special symplectic structures and Tulczyjew's triple. The
resulting construction describes discrete dynamics directly in terms
of Lagrangian submanifolds and Lagrangian relations, without requiring
the existence of a globally defined discrete evolution map. In this
way, regular, implicit, constrained, and degenerate discrete systems
can be incorporated into the same geometric framework.

A central role is played by Type--II discrete Morse families. A family
\[
F_k(\bq_k,\bp_{k+1},\lambda)
\]
generates a Lagrangian relation through the equations
\[
\bp_k
=
\frac{\partial F_k}{\partial\bq_k},
\qquad
\bq_{k+1}
=
\frac{\partial F_k}{\partial\bp_{k+1}},
\qquad
\frac{\partial F_k}{\partial\lambda}
=
0.
\]
The mixed variables $(\bq_k,\bp_{k+1})$ retain the natural Type--II
structure of discrete Hamiltonian mechanics, while the auxiliary
variables $\lambda$ encode the critical set required to describe
implicit dynamics. Ordinary Type--II discrete Hamiltonians are
recovered when no auxiliary variables are necessary.

This construction leads naturally to a discrete Hamilton--Jacobi
theory formulated directly in terms of the propagation of Lagrangian
submanifolds. Let
\[
E_k\subset T^*\Q\times T^*\Q
\]
denote the Lagrangian relation defining one step of the discrete
dynamics. If the Lagrangian submanifolds at two consecutive steps can
be represented as
\[
\Lambda_k=\operatorname{Im}\gamma_k,
\qquad
\Lambda_{k+1}=\operatorname{Im}\gamma_{k+1},
\]
for closed one-forms $\gamma_k$ and $\gamma_{k+1}$, the discrete
Hamilton--Jacobi condition is expressed geometrically as
\[
E_k\circ\Lambda_k
\subset
\Lambda_{k+1}.
\]
Thus, rather than identifying the incoming and outgoing momenta
through a single section, the formulation naturally distinguishes the
Lagrangian data at consecutive discrete steps.

When the one-forms are exact,
\[
\gamma_k=dW_k,
\qquad
\gamma_{k+1}=dW_{k+1},
\]
the preceding geometric condition gives a discrete
Hamilton--Jacobi relation between the consecutive generating
functions $W_k$ and $W_{k+1}$. This provides the usual generating
function description whenever the relevant Lagrangian submanifolds
are graphical.

The graphical description, however, is only local and may fail when
the projection onto configuration space becomes singular. To treat
this situation without selecting individual branches, we replace the
single generating function by a Morse family
\[
S_k(\bq_k,a)
\]
and propagate the corresponding Lagrangian submanifold directly.
Combining this family with the Type--II Morse family defining the
discrete dynamics leads to the composition family
\[
\mathcal S_{k+1}
(\bq_{k+1};
\bq_k,\bp_{k+1},a,\lambda)
=
S_k(\bq_k,a)
+
\langle\bp_{k+1},\bq_{k+1}\rangle
-
F_k(\bq_k,\bp_{k+1},\lambda).
\]
Criticality with respect to the internal variables recovers
simultaneously the Type--II discrete Hamilton equations, the
criticality equations of the initial Morse family, and the auxiliary
constraints contained in $F_k$. The resulting family therefore
generates the complete propagated Lagrangian submanifold without
requiring its projection onto configuration space to remain regular.

This observation is particularly relevant in the presence of
caustics. In Hamiltonian descriptions of geometrical optics and
semiclassical wave propagation, wavefronts are naturally associated
with Lagrangian manifolds generated by the Hamiltonian characteristic
flow \cite{MaslovFedoriuk1981,ArnoldCaustics1990}. At a caustic, the
projection of the Lagrangian manifold onto configuration space becomes
singular and the corresponding Hamilton--Jacobi description becomes
multivalued, although the underlying Lagrangian manifold may remain
regular \cite{ArnoldCaustics1990}. Morse families are therefore a
natural representation of the wavefront near such singularities,
since different branches are encoded by different critical points of
the same generating family
\cite{ArnoldGuseinVarchenko1985,ArnoldCaustics1990}.

We exploit this geometric feature to construct a discrete description
of wavefront propagation through fold caustics. Starting from an
optical Hamiltonian, we introduce a Type--II discrete Hamiltonian whose
discrete Hamilton equations define a symplectic ray integrator. The
wavefront is represented locally by a fold Morse family, and its
composition with the discrete Hamiltonian relation generates the
wavefront at the next discrete step. The corresponding criticality
conditions reproduce the discrete ray equations while retaining all
branches of the wavefront. Consequently, the discrete Hamiltonian
evolution remains regular through the formation of the caustic; the
singularity occurs in the projection of the propagated Lagrangian
submanifold onto configuration space rather than in the discrete
Hamiltonian dynamics itself.

\medskip

\noindent
\textbf{Standing hypothesis.} The mixed variables $(\bq_k,\bp_{k+1})$, the pairing
$\langle\bp_{k+1},\bq_{k+1}\rangle$ and the product $\Q\times\Q^*$ used throughout are
not canonical on a general manifold. We therefore assume, whenever Type--II objects
appear, that $\Q$ is a finite-dimensional real vector space, or an affine space with a
chosen origin, so that $T^*\Q\simeq\Q\times\Q^*$ globally. For a general configuration
manifold every such formula is to be read in a fixed local cotangent chart, and the
resulting statements are local.

\medskip

The main contributions of this work can be summarized as follows:
\begin{itemize}

\item We develop a discrete geometric framework based on Morse
families, Lagrangian relations, special symplectic structures, and
discrete analogues of Tulczyjew's triple.

\item We formulate implicit Type--II discrete Hamiltonian dynamics
directly in terms of Lagrangian relations generated by discrete Morse
families.

\item We formulate a discrete Hamilton--Jacobi theory as the
propagation of Lagrangian submanifolds between consecutive discrete
steps.

\item We extend this propagation principle from Lagrangian graphs to
general Lagrangian submanifolds represented by Morse families through
a composition of generating families.

\item We apply the construction to optical fold caustics, obtaining a
discrete propagation scheme that combines a symplectic ray integrator
with a generating-family representation of the multivalued wavefront.

\end{itemize}

The paper is organized as follows. Section~2 reviews the continuous
geometric structures used throughout the paper, including Lagrangian
submanifolds, Morse families, special symplectic structures,
Tulczyjew's triple, and Dirac systems. Section~3 develops their
discrete counterparts and introduces discrete generating families and
the discrete Tulczyjew triple. Section~4 formulates discrete Dirac
mechanics and discusses constrained and degenerate discrete systems.
Section~5 develops the Type--II discrete Hamilton--Jacobi theory and
the propagation of Lagrangian submanifolds represented by Morse
families. Finally, Section~6 applies the construction to the discrete
propagation of optical wavefronts through fold caustics.
\section{Continuous preliminaries}

\subsection{Symplectic geometry}
\noindent
Let $(\M,\omega)$ be a symplectic manifold and let $\N\subset\M$ be a submanifold.
The \emph{symplectic orthogonal complement} of $T\N$ is defined by
\[
T\N^\bot
:=
\left\{
u\in T\M
\;\middle|\;
\omega(u,v)=0
\text{ for all }
v\in T\N
\right\}.
\]

The symplectic orthogonal complement provides a natural way of comparing the tangent space of a submanifold with the ambient symplectic structure.
Several important classes of submanifolds are characterized in terms of the relation between $T\N$ and $T\N^\bot$.

\medskip

The submanifold $\N$ is called \emph{isotropic} if
\[
T\N\subset T\N^\bot.
\]
In this case,
\[
\dim\N\le \frac12\dim\M.
\]

The submanifold $\N$ is called \emph{coisotropic} if
\[
T\N^\bot\subset T\N.
\]
In this case,
\[
\dim\N\ge \frac12\dim\M.
\]

The submanifold $\N$ is called \emph{Lagrangian} if
\[
T\N=T\N^\bot.
\]
Equivalently, $\N$ is a maximal isotropic submanifold of $(\M,\omega)$.
In this case,
\[
\dim\N=\frac12\dim\M.
\]

Finally, $\N$ is called \emph{symplectic} if
\[
T\N\cap T\N^\bot=\{0\}.
\]
Equivalently, the restriction
\[
\omega_\N:=\iota^*\omega
\]
is non-degenerate and therefore defines a symplectic structure on $\N$. In particular
$\dim\N$ is then even, each tangent space $T_n\N$ being a symplectic vector space.

\medskip

These classes of submanifolds are preserved under symplectomorphisms.
In particular, the image of a Lagrangian (respectively isotropic, coisotropic, or symplectic) submanifold under a symplectomorphism is again Lagrangian (respectively isotropic, coisotropic, or symplectic).

\medskip

Two fundamental examples of Lagrangian submanifolds are:

\begin{itemize}
\item the fibers of the cotangent bundle projection
\[
\pi_\Q:T^*\Q\rightarrow \Q,
\]

\item the images of closed one-forms
\[
\gamma:\Q\rightarrow T^*\Q.
\]
\end{itemize}

The latter class includes, in particular, the zero section of the cotangent bundle.
More generally, by Weinstein's Lagrangian neighborhood theorem \cite{We77}, every Lagrangian submanifold admits a neighborhood symplectomorphic to a neighborhood of the zero section of a cotangent bundle.

\subsection{Morse families and special symplectic structures}
\label{subsec:morse-families-continuous}
Let
\[
(\Pp,\pi,\N)
\]
be a smooth fiber bundle.
A function
\[
F:\Pp\rightarrow\mathbb R
\]
may be regarded as a family of functions parametrized by the fibers of $\pi$.

The \emph{critical set} of $F$ relative to the fibration $\pi$ is defined by
\[
\Cr(F,\pi)
:=
\left\{
z\in\Pp
\;\middle|\;
\langle dF(z),V\rangle=0,
\quad
\forall V\in V_z\Pp
\right\},
\]
where
\[
V\Pp:=\ker(T\pi)
\]
denotes the vertical bundle of the fibration.

\medskip

For each
\[
z\in\Cr(F,\pi),
\]
define the bilinear map
\begin{align}
W(F,z):
V_z\Pp\times T_z\Pp
&\longrightarrow
\mathbb R,
\nonumber\\
(v,w)
&\longmapsto
D^{(1,1)}(F\circ\chi)(0,0),
\label{W}
\end{align}
where $D^{(1,1)}$ denotes the mixed second partial derivative,
\[
D^{(1,1)}(F\circ\chi)(0,0)
:=
\frac{\partial^2}{\partial s\,\partial t}(F\circ\chi)(0,0),
\]
the dependence of the right-hand side of \eqref{W} on the pair $(v,w)$ being carried
entirely by the choice of $\chi$, and where
\[
\chi:\mathbb R^2\rightarrow\Pp
\]
is any smooth map satisfying the tangent conditions at the origin,
$\chi(0,0)=z$,
\[
\left.\frac{\partial\chi}{\partial s}\right|_{(0,0)}=v,
\qquad
\left.\frac{\partial\chi}{\partial t}\right|_{(0,0)}=w,
\]
and, in addition,
\begin{equation}\label{eq:chi-vertical}
\pi\bigl(\chi(s,t)\bigr)=\pi\bigl(\chi(0,t)\bigr)
\end{equation}
near the origin, so that each curve $s\mapsto\chi(s,t)$ lies in a single fibre.

Condition \eqref{eq:chi-vertical} cannot be dispensed with. Expanding,
\[
D^{(1,1)}(F\circ\chi)(0,0)
=
D^2F(z)(v,w)
+
\bigl\langle dF(z),\partial_{st}\chi(0,0)\bigr\rangle ,
\]
and the second term is not determined by $v$ and $w$ alone: for
$\Pp=\mathbb R^2$ with $\pi(q,\lambda)=q$ and $F(q,\lambda)=q$, the maps
$\chi_1(s,t)=(t,s)$ and $\chi_2(s,t)=(t+st,s)$ both satisfy the first-derivative
conditions at $z=(0,0)$ with $v=\partial_\lambda$ and $w=\partial_q$, yet give
$D^{(1,1)}(F\circ\chi_1)=0$ and $D^{(1,1)}(F\circ\chi_2)=1$. Note that $z$ is a
critical point here, so criticality alone does not suffice. Under
\eqref{eq:chi-vertical} the vector $\partial_{st}\chi(0,0)$ is vertical and is
therefore annihilated by $dF(z)$ at a critical point, and $W(F,z)$ is well defined.
Equivalently and intrinsically, $W(F,z)$ is the normal derivative at $z$ of the
section $d_\pi F=dF|_{V\Pp}$ of $V^*\Pp$; in bundle coordinates $(q^i,\lambda^a)$,
\begin{equation}\label{eq:W-coord}
W(F,z)(v,w)
=
v^a\left(
\frac{\partial^2F}{\partial\lambda^a\partial q^i}\,w_q^i
+
\frac{\partial^2F}{\partial\lambda^a\partial\lambda^b}\,w_\lambda^b
\right).
\end{equation}

The family $F$ is said to be \emph{regular} if $\Cr(F,\pi)$ is a submanifold of $\Pp$
and
\[
\rank W(F,z)=\operatorname{codim}_\Pp\Cr(F,\pi)
\qquad\text{for every }z\in\Cr(F,\pi).
\]
Constancy of $\rank W$ along $\Cr(F,\pi)$ alone does not suffice for this, and does not
by itself make $\Cr(F,\pi)$ a submanifold of the corresponding codimension: for
$F(q,\lambda)=\lambda^4$ on $\Pp=\mathbb R^2$ with $\pi(q,\lambda)=q$, the critical set
$\{\lambda=0\}$ has codimension one while $W$ vanishes identically on it. One could instead impose a constant-rank hypothesis on the coordinate matrix
\eqref{MorseCon} in a neighbourhood of $\Cr(F,\pi)$, strong enough to invoke the
constant-rank theorem; note that this is a statement about a chart, since off the
critical set the rank of that matrix is not invariant under a change of bundle chart.

The family $F$ is called a \emph{Morse family} (or \emph{energy function}) if the rank of $W(F,z)$ is maximal for every
\[
z\in\Cr(F,\pi),
\]
that is, if $0$ is a regular value of the map $z\mapsto\partial F/\partial\lambda^a(z)$.
A Morse family is in particular regular, and since the maximal rank equals the fibre dimension of $\pi$, in that case
\[
\dim\Cr(F,\pi)=\dim\N .
\]

In local coordinates
\[
(q^i,\lambda^a)
\]
on $\Pp$, the critical set is determined by the equations
\[
\frac{\partial F}{\partial\lambda^a}=0.
\]

The family $F$ is a Morse family if and only if the Jacobian matrix of these equations with respect to the variables
\[
(q^i,\lambda^a)
\]
has maximal rank. Equivalently, the matrix
\begin{equation}\label{MorseCon}
\left(
\frac{\partial^2F}
{\partial\lambda^a\partial q^i}
\quad
\frac{\partial^2F}
{\partial\lambda^a\partial\lambda^b}
\right)
\end{equation}
has maximal rank.

\medskip

A Morse family $F$ on $(\Pp,\pi,\N)$ generates a Lagrangian submanifold
\begin{equation}\label{LagSub}
S_{T^*\N}
=
\left\{
w\in T^*\N
\;\middle|\;
(T_z^*\pi)(w)=dF(z)
\text{ for some }
z\in\Pp
\right\},
\end{equation}
with the compatibility condition
\[
\pi(z)=\pi_{T^*\N}(w).
\]

Here
\[
T_z^*\pi:
T_{\pi(z)}^*\N
\longrightarrow
T_z^*\Pp
\]
denotes the dual map of
\[
T_z\pi:
T_z\Pp
\longrightarrow
T_{\pi(z)}\N.
\]

The construction is summarized by the commutative diagram
\begin{equation}\label{Morse-pre}
\xymatrix{
\mathbb R
&
\Pp \ar[d]^{\pi} \ar[l]^{F}
&
T^*\N \ar[d]^{\pi_{T^*\N}}
\\
&
\N \ar@{=}[r]
&
\N .
}
\end{equation}

To describe $S_{T^*\N}$ more explicitly, define the map
\[
\kappa:
\Cr(F,\pi)
\longrightarrow
T^*\N
\]
by
\[
\langle \kappa(z),Z_\N(\pi(z))\rangle
=
\langle dF(z),Z_\Pp(z)\rangle,
\]
for every pair of $\pi$-related vector fields
\[
Z_\Pp\in\mathfrak X(\Pp),
\qquad
Z_\N\in\mathfrak X(\N),
\]
satisfying
\[
T\pi\circ Z_\Pp
=
Z_\N\circ\pi.
\]

The map $\kappa$ is an immersion and
\[
\dim S_{T^*\N}
=
\dim\Cr(F,\pi)
=
\dim\N.
\]
Its image coincides with the Lagrangian submanifold defined in \eqref{LagSub}.

\medskip

Now let $\N$ be an immersed submanifold of a manifold $\Q$.
Define
\[
T_\N^*\Q
=
\pi_\Q^{-1}(\N)
\subset T^*\Q.
\]

The map
\[
\xi:T_\N^*\Q\rightarrow T^*\N
\]
is defined by
\[
\langle\xi(p),Z_\N(n)\rangle
=
\langle p,Z_\N(n)\rangle,
\]
for all
\[
Z_\N\in\mathfrak X(\N),
\qquad
n=\pi_\Q(p).
\]

Let
\[
i:T_\N^*\Q\hookrightarrow T^*\Q
\]
denote the canonical inclusion.

If $S_{T^*\N}$ is a Lagrangian submanifold of $T^*\N$, then
\[
S_{T^*\Q}
=
i\circ\xi^{-1}(S_{T^*\N})
\]
is a Lagrangian submanifold of $T^*\Q$.

If $S_{T^*\N}$ is generated by a Morse family $F$, then we also say that $S_{T^*\Q}$ is generated by $F$.
This situation is represented by
\begin{equation}\label{Morse}
\xymatrix{
\mathbb R
&
\Pp \ar[d]^{\pi} \ar[l]^{F}
&
T^*\Q \ar[d]^{\pi_\Q}
\\
&
\N \ar@{^{(}->}[r]
&
\Q .
}
\end{equation}

In local coordinates $(q^i,\lambda^a)$, the generated Lagrangian submanifold is
\[
S_{T^*\Q}
=
\left\{
\left(
q^i,
\frac{\partial F}{\partial q^i}(q,\lambda)
\right)
\in T^*\Q
\;\middle|\;
\frac{\partial F}{\partial\lambda^a}(q,\lambda)=0
\right\}.
\]

This local formula is valid when $\N=\Q$. If $\N$ is a proper submanifold of $\Q$, the
fibres of $\xi:T_\N^*\Q\to T^*\N$ are the conormal spaces of $\N$, and these directions
must be retained: in adapted coordinates $(x,y)$ with $\N=\{y=0\}$,
\begin{equation}\label{eq:conormal-lift}
S_{T^*\Q}
=
\left\{
\left(x,0;\ p_x=\frac{\partial F}{\partial x}(x,\lambda),\ p_y=\mu\right)
\;\middle|\;
\frac{\partial F}{\partial\lambda^a}(x,\lambda)=0,
\ \mu\in\mathbb R^{\,\operatorname{codim}\N}
\right\},
\end{equation}
whose dimension is $\dim\Q$ as required. The multipliers $\mu$ play the role that Lagrange multipliers play in the constrained
systems of Sections~\ref{sec:discrete-morse-tulczyjew} and \ref{sec:discrete-dirac},
although there they arise as fibre variables of a Morse family over the full base rather
than as conormal directions of a smaller one. For an
immersion $j:\N\to\Q$ one works on $j^*T^*\Q$ with $\xi(n,p)=T^*_nj(p)$, and the object
generated is in general an immersed, not embedded, Lagrangian submanifold.

\medskip

Following Tulczyjew, a \emph{special symplectic structure} is a quintuple
$(\Rr,\N,\tau,\theta,A)$, where $\tau:\Rr\rightarrow\N$ is a fiber bundle,
$\theta\in\Omega^1(\Rr)$, and $A:\Rr\rightarrow T^*\N$ is a diffeomorphism satisfying
$\tau=\pi_\N\circ A$, $\theta=A^*\theta_\N$.

Since $(T^*\N,\omega_\N=-d\theta_\N)$ is symplectic, it follows that
$(\Rr,\omega=-d\theta)$ is also symplectic and $A^*\omega_\N=\omega$. Hence
$(\Rr,\omega)$ and $(T^*\N,\omega_\N)$ are symplectomorphic.

Let $(\Rr,\Q,\tau,\theta,A)$ be a special symplectic structure. If
$S_{T^*\Q}\subset T^*\Q$ is a Lagrangian submanifold, then $S=A^{-1}(S_{T^*\Q})$ is a
Lagrangian submanifold of $(\Rr,\omega)$.

If $S_{T^*\Q}$ is generated by a Morse family $F$ as in Diagram~\eqref{Morse}, then $S$
is also said to be generated by $F$. This situation is summarized by
\begin{equation}\label{Morse-Gen}
\xymatrix{
\mathbb R & \Pp \ar[d]^{\pi} \ar[l]^{F} & T^*\Q \ar[d]_{\pi_\Q} & \Rr \ar[l]_{A} \ar[dl]^{\tau} \\
& \N \ar@{^{(}->}[r] & \Q .
}
\end{equation}

We shall return to these constructions in the next section, where we introduce a
fundamental example of a special symplectic structure, namely the Tulczyjew symplectic
space $TT^*\Q$.

\subsection{Geometry of mechanical bundles}

Let $\Q$ be a smooth manifold. We denote by $\pi_\Q:T^*\Q\rightarrow\Q$ the canonical
projection of the cotangent bundle and by $\tau_\Q:T\Q\rightarrow\Q$ the canonical
projection of the tangent bundle.

Local coordinates on $T^*\Q$ are denoted by $(q^i,p_i)$, where $\pi_\Q(q^i,p_i)=q^i$,
while local coordinates on $T\Q$ are denoted by $(q^i,\dot q^i)$, where
$\tau_\Q(q^i,\dot q^i)=q^i$.

The cotangent bundle $T^*\Q$ carries the canonical one-form
$\theta_\Q\in\Omega^1(T^*\Q)$, defined by
\[
\theta_\Q(\alpha_q)(X_{\alpha_q})=\alpha_q\bigl(T\pi_\Q(X_{\alpha_q})\bigr),
\qquad
X_{\alpha_q}\in T_{\alpha_q}(T^*\Q),
\quad
\alpha_q\in T_q^*\Q.
\]
In local coordinates,
$\theta_\Q=p_i\,dq^i$, which is the Liouville one-form. The canonical symplectic form on
$T^*\Q$ is $\omega_\Q=-d\theta_\Q=dq^i\wedge dp_i$. Since $\omega_\Q$ is closed and
non-degenerate, $(T^*\Q,\omega_\Q)$ is a symplectic manifold.

Given two symplectic manifolds $(\M_1,\omega_1)$, $(\M_2,\omega_2)$, a diffeomorphism
$F:\M_1\rightarrow\M_2$ is called a symplectomorphism if $F^*\omega_2=\omega_1$.

We now turn to the tangent bundle $T\Q$. For a smooth function
$f:\Q\rightarrow\mathbb R$, its complete lift $f^T:T\Q\rightarrow\mathbb R$ is defined by
$f^T(v_q)=df(q)(v_q)$, $v_q\in T_q\Q$. In local coordinates,
$f^T(q,\dot q)=\dot q^i\,\partial f/\partial q^i$.

Given two tangent vectors $v_q,w_q\in T_q\Q$, the vertical lift of $v_q$ at $w_q$ is
defined by $(v_q)^V_{w_q}=\left.\frac{d}{dt}\right|_{t=0}(w_q+t\,v_q)$. If
$X\in\mathfrak X(\Q)$ is a vector field on $\Q$, its vertical lift to $T\Q$ is denoted by
$X^V$.

Let $\phi_t$ be the flow of $X$. The complete lift of $X$ is the infinitesimal generator
of the tangent lift $T\phi_t:T\Q\rightarrow T\Q$. If
$X=X^i(q)\,\partial/\partial q^i$, then its complete lift is
$X^T=X^i(q)\,\partial/\partial q^i+\dot q^j\,(\partial X^i/\partial q^j)\,
\partial/\partial\dot q^i$.

We now consider the tangent bundle of the cotangent bundle, $TT^*\Q$. The induced
coordinates on $TT^*\Q$ are denoted by $(q^i,p_i,\dot q^i,\dot p_i)$.

The complete lift of the canonical symplectic form $\omega_\Q$ to $TT^*\Q$ determines a
symplectic form $\Omega_\Q=-d_T\omega_\Q$, known as the Tulczyjew symplectic form. The
sign records the convention $\omega_\Q=-d\theta_\Q=dq^i\wedge dp_i$ fixed above, for
which $d_T\omega_\Q=d\dot q^i\wedge dp_i+dq^i\wedge d\dot p_i$; with Tulczyjew's
convention $\omega_\Q=d\theta_\Q=dp_i\wedge dq^i$ one has instead
$\Omega_\Q=d_T\omega_\Q$.

A remarkable feature of this symplectic structure is that it admits two globally defined
symplectic potentials, denoted by $\theta_1$, $\theta_2$. These are defined intrinsically
by $\theta_1=i_T\omega_\Q$, $\theta_2=d_T\theta_\Q$. Here the subscript $T$ refers to the
tautological vector field along $\tau_{T^*\Q}:TT^*\Q\to T^*\Q$, that is, to the
assignment sending a point $v\in TT^*\Q$ to the tangent vector $v$ itself: for a
$k$-form $\beta$ on $T^*\Q$, the contraction $i_T\beta$ is the $(k-1)$-form on $TT^*\Q$
given by
\begin{equation}\label{eq:iT}
(i_T\beta)_v(X_1,\dots,X_{k-1})
=
\beta_{\tau_{T^*\Q}(v)}
\bigl(v,\;T\tau_{T^*\Q}(X_1),\dots,T\tau_{T^*\Q}(X_{k-1})\bigr),
\end{equation}
and $d_T=[d,i_T]=d\circ i_T+i_T\circ d$ is the associated operator, the graded bracket
being an anticommutator because $i_T$ lowers the degree by one. Both $i_T$ and $d_T$ send
forms on $T^*\Q$ to forms on $TT^*\Q$, so $d_T$ is a degree-zero derivation \emph{along}
$\tau_{T^*\Q}^*$, in the sense that
\[
d_T(\beta\wedge\gamma)
=
d_T\beta\wedge\tau_{T^*\Q}^*\gamma
+
\tau_{T^*\Q}^*\beta\wedge d_T\gamma .
\]
It is precisely the complete, or tangent, lift of forms, and it extends the complete lift
of functions introduced above for $f:\Q\to\mathbb R$ to any $f:T^*\Q\to\mathbb R$, since
$d_Tf=i_T(df)$ is the derivative of $f$ along the tautological vector, which we again
write $f^T$.

In local coordinates on $TT^*\Q$, these one-forms are given by
\begin{equation}\label{thets}
\theta_1=-\dot p_i\,dq^i+\dot q^i\,dp_i,
\qquad
\theta_2=\dot p_i\,dq^i+p_i\,d\dot q^i.
\end{equation}
Their exterior derivatives differ by a sign, $d\theta_1=-d\theta_2$, and define the
canonical symplectic form on $TT^*\Q$:
\begin{equation}\label{Tstf}
\Omega_\Q=-d\theta_1=d\theta_2=d\dot p_i\wedge dq^i+dp_i\wedge d\dot q^i.
\end{equation}

\subsection{Tulczyjew's triple}

The Tulczyjew symplectic manifold $(TT^{*}\Q,\Omega_\Q)$ admits two natural special
symplectic structures.

The non-degeneracy of the canonical symplectic form $\omega_\Q$ induces the musical
isomorphism
\begin{equation}\label{beta1}
\beta_\Q:TT^{*}\Q\longrightarrow T^{*}T^{*}\Q,
\qquad
X\longmapsto \iota_X\omega_\Q .
\end{equation}
In \eqref{beta1} the letter $X$ denotes a point of $TT^{*}\Q$, that is, a tangent vector
$X\in T_\alpha(T^*\Q)$ at some $\alpha\in T^*\Q$, and $\omega_\Q$ is the symplectic form
on $T^*\Q$; the contraction $\iota_X\omega_\Q$ is therefore an element of
$T^*_\alpha(T^*\Q)$, as the stated target requires. Contracting instead with the
Tulczyjew form $\Omega_\Q$ would produce a covector on $TT^*\Q$ rather than on $T^*\Q$,
and would not define a map into $T^{*}T^{*}\Q$.

With the convention $\omega_\Q=-d\theta_\Q$ adopted above, $\beta_\Q$ is a
symplectomorphism onto $T^{*}T^{*}\Q$ equipped with its canonical symplectic structure,
since $\beta_\Q^*\theta_{T^*T^*\Q}=\theta_1$ and hence
$\beta_\Q^*\omega_{T^*T^*\Q}=-d\theta_1=\Omega_\Q$. In local coordinates,
\begin{equation}\label{beta2}
\beta_\Q(q^i,p_i,\dot q^i,\dot p_i)=(q^i,p_i,-\dot p_i,\dot q^i).
\end{equation}
Consequently, $(TT^{*}\Q,T^{*}\Q,\tau_{T^{*}\Q},\theta_1,\beta_\Q)$ is a special
symplectic structure.

\medskip

The second special symplectic structure is induced by the canonical involution on
$TT\Q$ and is given by Tulczyjew's isomorphism
$\alpha_\Q:TT^{*}\Q\longrightarrow T^{*}T\Q$. With the convention
$\omega_\Q=-d\theta_\Q$ adopted above, this map is \emph{anti}-symplectic for
$\Omega_\Q$: since $\alpha_\Q^*\theta_{T\Q}=\theta_2$ one has
$\alpha_\Q^*\omega_{T^*T\Q}=-d\theta_2=-\Omega_\Q$. Equivalently, $\alpha_\Q$ is a
symplectomorphism from $(TT^*\Q,-\Omega_\Q)$, and the special symplectic structure it
defines carries $-\Omega_\Q$. Under Tulczyjew's opposite convention
$\omega_\Q=d\theta_\Q=dp_i\wedge dq^i$ both wings are symplectic for the same form,
although $\beta_\Q$ then reads $(q^i,p_i,\dot q^i,\dot p_i)\mapsto(q^i,p_i,\dot p_i,-\dot
q^i)$ rather than \eqref{beta2}. The asymmetry is a matter of convention only, and is
immaterial for what follows, since a submanifold is Lagrangian for $\Omega$ if and only
if it is Lagrangian for $-\Omega$. In local coordinates,
\begin{equation}\label{defalpha}
\alpha_\Q(q^i,p_i,\dot q^i,\dot p_i)=(q^i,\dot q^i,\dot p_i,p_i).
\end{equation}
Hence, $(TT^{*}\Q,T\Q,T\pi_\Q,\theta_2,\alpha_\Q)$ is also a special symplectic
structure.

\medskip

Together, these two special symplectic structures define Tulczyjew's triple:
\begin{equation}\label{T}
\xymatrix{
T^{*}T\Q
\ar[dr]^{\pi_{T\Q}}
&&
TT^{*}\Q
\ar[dl]_{T\pi_\Q}
\ar[dr]^{\tau_{T^{*}\Q}}
\ar[rr]^{\beta_\Q}
\ar[ll]_{\alpha_\Q}
&&
T^{*}T^{*}\Q
\ar[dl]^{\pi_{T^{*}\Q}}
\\
&
T\Q
&&
T^{*}\Q
}
\end{equation}

\subsubsection{Lagrangian submanifolds of Tulczyjew's space}

In the geometric framework adopted in this paper, implicit Hamiltonian systems are represented by Lagrangian submanifolds of the symplectic manifold
\[
(TT^{*}\Q,\Omega_\Q).
\]

Let
\[
E\subset TT^{*}\Q
\]
be a submanifold defined locally by constraint functions
\[
\Phi^{A}:TT^{*}\Q\rightarrow\mathbb R,
\qquad
\Phi^{A}(q,p,\dot q,\dot p)=0.
\]

Since
\[
\dim TT^{*}\Q=4n,
\]
a Lagrangian submanifold has dimension \(2n\). Therefore, if \(E\) is described as a regular constraint submanifold, it must be locally defined by \(2n\) independent equations. Moreover, the defining constraints must be in involution on \(E\) with respect to the Poisson bracket induced by \(\Omega_\Q\) \cite{LiMa}.

The Poisson bracket associated with the symplectic form \eqref{Tstf} is
\begin{equation}\label{PB}
\{f,g\}
=
\frac{\partial f}{\partial \dot p_i}
\frac{\partial g}{\partial q^i}
-
\frac{\partial g}{\partial \dot p_i}
\frac{\partial f}{\partial q^i}
+
\frac{\partial f}{\partial p_i}
\frac{\partial g}{\partial \dot q^i}
-
\frac{\partial g}{\partial p_i}
\frac{\partial f}{\partial \dot q^i}.
\end{equation}

The image of a Hamiltonian vector field on \(T^*\Q\) is a Lagrangian submanifold of \(TT^*\Q\). Conversely, if a Lagrangian submanifold
\[
E\subset TT^*\Q
\]
is locally the image of a vector field
\[
X\in\mathfrak X(T^*\Q),
\]
then \(X\) is locally Hamiltonian.

Assume that
\[
E=\operatorname{Im}(X),
\]
where
\[
X
=
\phi^i(q,p)\frac{\partial}{\partial q^i}
+
\psi_i(q,p)\frac{\partial}{\partial p_i}.
\]

Then \(E\) is locally given by
\begin{equation}\label{Eedef}
E
=
\left\{
(q,p,\dot q,\dot p)\in TT^*\Q
\;\middle|\;
\dot q^i-\phi^i(q,p)=0,
\qquad
\dot p_i-\psi_i(q,p)=0
\right\}.
\end{equation}

Since \(E\) is Lagrangian, the Poisson brackets of the defining constraint functions in \eqref{Eedef} vanish on \(E\). Written out, these integrability conditions are
\begin{equation}\label{eq:integrability}
\frac{\partial\phi^i}{\partial p_j}
=
\frac{\partial\phi^j}{\partial p_i},
\qquad
\frac{\partial\psi_i}{\partial q^j}
=
\frac{\partial\psi_j}{\partial q^i},
\qquad
\frac{\partial\psi_i}{\partial p_j}
=
-\frac{\partial\phi^j}{\partial q^i},
\end{equation}
and they are precisely the conditions for the one-form
\[
\phi
=
-\psi_i\,dq^i
+
\phi^i\,dp_i
=
\iota_X\omega_\Q
=
\beta_\Q(X)
\]
to be closed. Intrinsically, therefore, \(E=\operatorname{Im}(X)\) is Lagrangian if and only if the one-form \(\beta_\Q(X)\) on \(T^*\Q\) is closed.

Therefore, at least locally, there exists a function
\[
H:T^*\Q\rightarrow\mathbb R
\]
such that
\[
dH=\phi.
\]

Consequently,
\[
\phi^i
=
\frac{\partial H}{\partial p_i},
\qquad
\psi_i
=
-\frac{\partial H}{\partial q^i},
\]
and the vector field \(X\) is Hamiltonian.

\medskip

By the Maslov--H\"ormander theorem
\cite{BeCaSi09,Be11,Ca91,Ca15},
every Lagrangian submanifold
\[
E\subset TT^*\Q
\]
can be generated locally by a Morse family.

This situation is represented by the diagram
\begin{equation}\label{Morse-Tul}
\xymatrix{
\mathbb R
&
\Pp \ar[d]^{\pi} \ar[l]^{F}
&
T^*T^*\Q \ar[d]_{\pi_{T^*\Q}}
&
TT^*\Q \ar[l]_{\beta_\Q} \ar[dl]^{\tau_{T^*\Q}}
\\
&
\N \ar@{^{(}->}[r]
&
T^*\Q .
}
\end{equation}

Locally, the generated Lagrangian submanifold is
\begin{equation}\label{MFGen}
E
=
\left\{
\left(
q^i,p_i,
\frac{\partial F}{\partial p_i},
-\frac{\partial F}{\partial q^i}
\right)
\in TT^*\Q
\;\middle|\;
\frac{\partial F}{\partial \lambda^a}(q,p,\lambda)=0
\right\}.
\end{equation}

\subsection{Dirac structures and mechanics}

We follow the treatment of Dirac structures in Lagrangian mechanics developed by
Yoshimura and Marsden \citep{YoshimuraMarsden2006a,YoshimuraMarsden2006b}, whose
discrete counterpart is due to Leok and Ohsawa \citep{leok2011variational}.

\subsubsection{The induced Dirac structure on $T^*\Q$}

Let
\[
\pi_\Q:T^*\Q\rightarrow\Q
\]
denote the canonical cotangent bundle projection, and let
\[
\beta_\Q:TT^*\Q\rightarrow T^*T^*\Q
\]
be the right wing of Tulczyjew's triple associated with the canonical symplectic structure on \(T^*\Q\).

Let
\[
\Delta_\Q\subset T\Q
\]
be a distribution of constant rank.
Its lift to \(T^*\Q\) is defined by
\[
\Delta_{T^*\Q}
:=
(T\pi_\Q)^{-1}(\Delta_\Q)
\subset TT^*\Q.
\]

Denote by
\[
\Delta_\Q^\circ\subset T^*\Q
\]
the annihilator of \(\Delta_\Q\), and by
\[
\Delta_{T^*\Q}^\circ
=
\pi_\Q^*(\Delta_\Q^\circ)
\subset T^*T^*\Q
\]
the annihilator of \(\Delta_{T^*\Q}\).

The vector subbundle
\begin{equation}\label{eq:Dirac-induced}
D_{\Delta_\Q}
=
\left\{
(v,\alpha)\in TT^*\Q\oplus T^*T^*\Q
\;\middle|\;
v\in\Delta_{T^*\Q},
\quad
\alpha-\beta_\Q(v)\in\Delta_{T^*\Q}^\circ
\right\}
\end{equation}
defines a Dirac structure on \(T^*\Q\).
It is called the \emph{induced Dirac structure} associated with the distribution \(\Delta_\Q\).

\subsubsection{Dirac structures and mechanics}

Let
\[
L:T\Q\rightarrow\mathbb R
\]
be a Lagrangian function.
The \emph{Dirac differential} of \(L\) is the mapping
\[
\mathfrak{D}L:T\Q\rightarrow T^*T^*\Q,
\qquad
\mathfrak{D}L:=\gamma_\Q\circ dL,
\]
where
\[
\gamma_\Q:T^*T\Q\rightarrow T^*T^*\Q
\]
is the canonical diffeomorphism induced by Tulczyjew's construction.

In local coordinates,
\[
\mathfrak{D}L(q^i,v^i)
=
\left(
q^i,
\frac{\partial L}{\partial v^i},
-\frac{\partial L}{\partial q^i},
v^i
\right).
\]

Let
\[
D\subset TT^*\Q\oplus T^*T^*\Q
\]
be a Dirac structure on \(T^*\Q\), and let
\[
X\in\mathfrak X(T^*\Q)
\]
be a vector field.

A \emph{Lagrange--Dirac system} is defined by the condition
\begin{equation}\label{eq:LDS}
(X,\mathfrak{D}L)\in D,
\end{equation}
understood at a common base point. This is a genuine restriction, not a formality:
$\mathfrak DL(q,v)$ is a covector based at $\bigl(q,\partial L/\partial v\bigr)$, whereas
$X$ is a vector based at an independently chosen $(q,p)$, and for degenerate $L$ the
velocity cannot be recovered from $(q,p)$. The condition is therefore imposed on the
fibre product
\[
\mathcal K
=
\left\{(q,v,p)\in T\Q\oplus_\Q T^*\Q
\;\middle|\;
p=\frac{\partial L}{\partial v}(q,v)\right\},
\]
on which one asks that
$\bigl(X(q,v,p),\mathfrak DL(q,v)\bigr)\in D(q,p)$ for a partial vector field $X$ along
$\mathcal K$. With this understood, the coordinate equations below are correct as
written.

In particular, if
\[
D=D_{\Delta_\Q},
\]
then \eqref{eq:LDS} becomes
\[
T\pi_\Q(X)\in\Delta_\Q,
\qquad
\beta_\Q(X)-\mathfrak{D}L
\in
\Delta_{T^*\Q}^\circ.
\]

Writing locally
\[
X
=
\dot q^i\frac{\partial}{\partial q^i}
+
\dot p_i\frac{\partial}{\partial p_i},
\]
one obtains
\begin{equation}\label{eq:LDS-coord}
\dot q^i\frac{\partial}{\partial q^i}\in\Delta_\Q(q),
\qquad
\dot q^i=v^i,
\qquad
p_i=\frac{\partial L}{\partial v^i},
\qquad
\left(\dot p_i-\frac{\partial L}{\partial q^i}\right)dq^i
\in
\Delta_\Q^\circ(q).
\end{equation}

These equations define the nonholonomic Lagrange--Dirac dynamics associated with the pair \((L,\Delta_\Q)\).

\medskip

Implicit Hamiltonian systems are defined in a completely analogous manner.

Let
\[
H:T^*\Q\rightarrow\mathbb R
\]
be a Hamiltonian function and let
\[
D\subset TT^*\Q\oplus T^*T^*\Q
\]
be a Dirac structure.

An implicit Hamiltonian system is specified by the condition
\[
(X,dH)\in D,
\qquad
X\in\mathfrak X(T^*\Q).
\]

When
\[
D=D_{\Delta_\Q},
\]
this condition becomes
\[
T\pi_\Q(X)\in\Delta_\Q,
\qquad
\beta_\Q(X)-dH
\in
\Delta_{T^*\Q}^\circ.
\]

In local coordinates, one obtains the nonholonomic Hamilton equations
\citep{BaSn1993,ScMa1994}
\[
\dot q^i\frac{\partial}{\partial q^i}\in\Delta_\Q(q),
\qquad
\dot q^i
=
\frac{\partial H}{\partial p_i},
\qquad
\left(\dot p_i
+
\frac{\partial H}{\partial q^i}\right)dq^i
\in
\Delta_\Q^\circ(q).
\]

When
\[
\Delta_\Q=T\Q,
\]
the annihilator \(\Delta_\Q^\circ\) vanishes and the above equations reduce to the standard Hamilton equations
\[
\dot q^i
=
\frac{\partial H}{\partial p_i},
\qquad
\dot p_i
=
-\frac{\partial H}{\partial q^i}.
\]

\section{Discrete preliminaries}

\subsection{Discrete variational mechanics}

Let $\Q$ be a configuration manifold.
A discrete Lagrangian is a smooth function
\[
L_d:\Q\times\Q\rightarrow\mathbb R.
\]

Given a discrete path
\[
\{\bq_k\}_{k=0}^{N},
\]
the associated discrete action is
\[
G_d
=
\sum_{k=0}^{N-1}
L_d(\bq_k,\bq_{k+1}).
\]

Hamilton's principle for the discrete action, with fixed endpoints, yields the discrete Euler--Lagrange equations
\[
D_2L_d(\bq_{k-1},\bq_k)
+
D_1L_d(\bq_k,\bq_{k+1})
=
0.
\]

The discrete Legendre transforms are defined by
\[
F^+_{L_d}(\bq_k,\bq_{k+1})
=
\bigl(
\bq_{k+1},
D_2L_d(\bq_k,\bq_{k+1})
\bigr),
\]
and
\[
F^-_{L_d}(\bq_k,\bq_{k+1})
=
\bigl(
\bq_k,
-D_1L_d(\bq_k,\bq_{k+1})
\bigr).
\]

If both maps are local diffeomorphisms, the discrete Lagrangian is said to be regular.
In that case, the discrete Euler--Lagrange equations determine a symplectic evolution map on $T^*\Q$.

When either discrete Legendre transform fails to be locally invertible, the discrete Lagrangian is degenerate and the resulting dynamics becomes implicit. The purpose of the present work is precisely to develop a geometric framework for such implicit discrete systems using discrete generating families, discrete Tulczyjew triples, and discrete Dirac structures.

The discrete variational formulation reviewed above describes regular systems through discrete Euler--Lagrange equations and symplectic evolution maps. In order to treat degenerate Lagrangians, implicit dynamics, and constrained systems, additional geometric structures are required.

The aim of this section is to develop discrete counterparts of Morse families, special symplectic structures, and Tulczyjew's triple. These constructions provide the geometric framework underlying the discrete Dirac and Hamilton--Jacobi theories developed in the subsequent sections.

\subsection{Discrete generating families}

Let
\[
(\Pp_d,\pi_d,\N_d)
\]
be a fibration between discrete configuration spaces, regarded as a discretization of a smooth bundle
\[
(\Pp,\pi,\N).
\]

A smooth function
\[
F:\Pp_d\rightarrow\mathbb R
\]
is called a discrete generating family.

The \emph{discrete critical set} of $F$ relative to the fibration $\pi_d$ is defined by
\[
\Cr(F,\pi_d)
:=
\left\{
z\in\Pp_d
\;\middle|\;
\langle dF(z),V\rangle=0,
\quad
\forall V\in V_z\Pp_d
\right\},
\]
where
\[
V\Pp_d:=\ker(T\pi_d)
\]
denotes the vertical bundle of the discrete fibration.

So that the base is unambiguous, we fix once and for all a base manifold $\N_d$ with
local coordinates $x^A$ and write $(x^A,\lambda^a)$ for the induced bundle coordinates
on $\Pp_d$. Three bases occur below: $\N_d=\Q$, used in the remainder of this subsection, in the
discrete special symplectic structures of Section~\ref{sec:disc-sss}, and for the
wavefront families of Sections~\ref{sec:composition-morse}--\ref{sec:reduction} and
\ref{sec:optical-application};
$\N_d=\Q\times\Q$, with $x=(q_k,q_{k+1})$, for Type--I families; and
$\N_d=\Q\times\Q^*$, with $x=(q_k,p_{k+1})$, for Type--II families. In each case the
generated object \eqref{eq:LF} lives in the cotangent bundle of \emph{that} base, so
that the diagram \eqref{eq:disc-morse} below depicts the first case only.

In these coordinates the discrete critical set is determined by the equations
\[
\frac{\partial F}{\partial\lambda^a}=0 .
\]

The generating family $F$ is a \emph{discrete Morse family} if and only if the Jacobian
matrix of these equations with respect to \emph{all} the variables $(x^A,\lambda^a)$ has
maximal rank $m_\lambda$, that is, if and only if
\begin{equation}\label{eq:disc-gen-rank}
\left(
\frac{\partial^2F}
{\partial\lambda^a\partial x^A}
\quad
\frac{\partial^2F}
{\partial\lambda^a\partial\lambda^b}
\right)
\end{equation}
has maximal rank. Every base coordinate must appear here; omitting the $q_{k+1}$ or
$p_{k+1}$ columns gives a strictly stronger and incorrect condition.

In that case,
\[
\dim\Cr(F,\pi_d)
=
\dim\N_d ,
\]
and, as in the continuous theory, $F$ generates an immersed Lagrangian submanifold of
the cotangent bundle \emph{of that base},
\begin{equation}\label{eq:LF}
\begin{aligned}
L_F
&=
\left\{
\eta\in T^*\N_d
\;\middle|\;
(T_z^*\pi_d)\,\eta=dF(z)
\ \text{for some }z\in\Cr(F,\pi_d)\text{ with }\pi_d(z)=\pi_{T^*\N_d}(\eta)
\right\}
\\
&=
\left\{
\left(x,\frac{\partial F}{\partial x}(x,\lambda)\right)
\;\middle|\;
\frac{\partial F}{\partial\lambda^a}(x,\lambda)=0
\right\}.
\end{aligned}
\end{equation}

In the special case $\N_d=\Q$, and using the standing identification
$T^*\Q\simeq\Q\times\Q^*$, \eqref{eq:LF} reads
\begin{equation}\label{eq:disc-lag-sub}
S_{\Q\times\Q^*}
=
\left\{
\left(
q^i,
\frac{\partial F}{\partial q^i}(q,\lambda)
\right)
\in
\Q\times\Q^*
\;\middle|\;
\frac{\partial F}{\partial\lambda^a}(q,\lambda)=0
\right\}.
\end{equation}
If the intended base is a proper submanifold of $\Q$, the conormal directions of
\eqref{eq:conormal-lift} must be adjoined as in the continuous case.

The construction is represented by the diagram
\begin{equation}\label{eq:disc-morse}
\xymatrix{
\mathbb R
&
\Pp_d \ar[d]^{\pi_d} \ar[l]^{F}
&
\Q\times\Q^* \ar[d]^{\pi_\Q}
\\
&
\N_d \ar@{^{(}->}[r]
&
\Q .
}
\end{equation}

Discrete generating families provide a unified framework for describing explicit and implicit discrete dynamics, including systems arising from degenerate discrete Lagrangians and constrained variational principles.

\subsection{Discrete special symplectic structures}
\label{sec:disc-sss}

The continuous notion of a special symplectic structure is based on a symplectomorphism with a cotangent bundle. In the discrete setting we adopt the same philosophy, replacing cotangent bundles by the discrete phase space $\Q\times\Q^*$.

A discrete special symplectic structure consists of a quadruple
\[
\bigl((\Rr,\Omega_\Rr),\N_d,\tau,A_d\bigr),
\]
where $\N_d$ is the base of the generating family,
\[
\tau:\Rr\rightarrow\N_d
\]
is a surjective submersion and
\[
A_d:\Rr\rightarrow T^*\N_d
\]
is a diffeomorphism satisfying
\[
\tau=\pi_{\N_d}\circ A_d,
\qquad
A_d^*\,\omega_{\N_d}=\Omega_\Rr .
\]
The symplectic form on $\Rr$ is thus part of the data, or equivalently is \emph{defined}
by the last identity; without it the requirement that $A_d$ be symplectic would be
circular, since $\Rr$ carries no a priori symplectic structure. For $\N_d=\Q$ the
target is $T^*\Q\simeq\Q\times\Q^*$ under the standing vector-space hypothesis, and
$\omega_{\N_d}$ is the canonical form $dq^i\wedge dp_i$ transported by that
identification; the two wings of the discrete Tulczyjew triple constructed in the next
subsection are the instances $\N_d=\Q\times\Q$ and $\N_d=\Q\times\Q^*$, up to the sign
of the target form recorded there.

If
\[
L_F
\subset
T^*\N_d
\]
is the Lagrangian submanifold \eqref{eq:LF} generated by a discrete generating family
$F$, then its inverse image
\[
S=A_d^{-1}(L_F)
\]
is a Lagrangian submanifold of $\Rr$.

In this case we also say that $S$ is generated by the discrete generating family $F$.

The situation is summarized by the diagram
\begin{equation}\label{eq:disc-morse-gen}
\xymatrix{
\mathbb R
&
\Pp_d \ar[d]^{\pi_d} \ar[l]^{F}
&
T^*\N_d \ar[d]_{\pi_{\N_d}}
&
\Rr \ar[l]_{A_d} \ar[dl]^{\tau}
\\
&
\N_d \ar@{=}[r]
&
\N_d .
}
\end{equation}

These constructions provide the geometric setting needed to introduce discrete analogues of Tulczyjew's triple and discrete implicit Hamiltonian systems.

\subsection{Geometry of discrete mechanical bundles}

\subsubsection{Discrete Tulczyjew's triple}

We now construct a discrete analogue of Tulczyjew's triple. In contrast with the continuous theory, where the Tulczyjew triple is formulated in terms of symplectic manifolds and vector bundle morphisms, the discrete version is described by symplectic maps generated by discrete generating functions.

Throughout this section we assume that $\Q$ is a finite-dimensional vector space, so that
\[
T^*\Q \simeq \Q\times\Q^*.
\]

The discrete counterpart of the Tulczyjew symplectic space $TT^*\Q$ is therefore taken to be
\[
T^*\Q\times T^*\Q.
\]

The discrete Tulczyjew triple relates the spaces
\[
T^*(\Q\times\Q),
\qquad
T^*\Q\times T^*\Q,
\qquad
T^*(\Q\times\Q^*)
\]
through diffeomorphisms associated with Type I and Type II generating functions, whose
behaviour with respect to the canonical forms on the three spaces is recorded after
\eqref{eq:disc-Omega} below.

Let
\[
F:T^*\Q\rightarrow T^*\Q,
\qquad
(q_0,p_0)\mapsto(q_1,p_1),
\]
be a symplectic map.

Associated with $F$ are the maps
\begin{align*}
F_1 &: \Q\times\Q \longrightarrow \Q^*\times\Q^*,
&
(q_0,q_1) &\longmapsto (p_0,p_1),
\\[2mm]
F_2 &: \Q\times\Q^* \longrightarrow \Q^*\times\Q,
&
(q_0,p_1) &\longmapsto (p_0,q_1).
\end{align*}

\paragraph{Type I generating functions}

Assume that $(p_0,p_1)$ can be expressed locally as functions of $(q_0,q_1)$.

The symplecticity condition implies
\[
d\!\left(
-p_{0i}\,dq_0^i
+
p_{1i}\,dq_1^i
\right)=0.
\]

Hence, by the Poincar\'e lemma, there exists locally a function
\[
S_1:\Q\times\Q\rightarrow\mathbb R
\]
such that
\[
-p_{0i}\,dq_0^i
+
p_{1i}\,dq_1^i
=
dS_1(q_0,q_1).
\]

Therefore
\begin{equation}\label{eq:S1}
p_{0i}
=
-\frac{\partial S_1}{\partial q_0^i},
\qquad
p_{1i}
=
\frac{\partial S_1}{\partial q_1^i}.
\end{equation}

This determines the diffeomorphism
\[
\alpha_\Q^{d}:
T^*\Q\times T^*\Q
\longrightarrow
T^*(\Q\times\Q),
\]
given locally by
\[
\alpha_\Q^{d}
\big((q_0,p_0),(q_1,p_1)\big)
=
(q_0,q_1,-p_0,p_1).
\]

\paragraph{Type II generating functions}

Assume now that $(p_0,q_1)$ can be expressed locally as functions of $(q_0,p_1)$.

The symplecticity condition becomes
\[
d\!\left(
p_{0i}\,dq_0^i
+
q_1^i\,dp_{1i}
\right)=0.
\]

Consequently there exists locally a generating function
\[
S_2:\Q\times\Q^*\rightarrow\mathbb R
\]
such that
\[
p_{0i}\,dq_0^i
+
q_1^i\,dp_{1i}
=
dS_2(q_0,p_1).
\]

Hence
\begin{equation}\label{eq:S2}
p_{0i}
=
\frac{\partial S_2}{\partial q_0^i},
\qquad
q_1^i
=
\frac{\partial S_2}{\partial p_{1i}}.
\end{equation}

This determines the diffeomorphism
\[
\beta_\Q^{d+}:
T^*\Q\times T^*\Q
\longrightarrow
T^*(\Q\times\Q^*),
\]
given by
\[
\beta_\Q^{d+}
\big((q_0,p_0),(q_1,p_1)\big)
=
(q_0,p_1,p_0,q_1).
\]

The maps
\[
\alpha_\Q^d
\qquad\text{and}\qquad
\beta_\Q^{d+}
\]
constitute the $(+)$-discrete Tulczyjew triple.

$Discrete\; symplectic\; forms.$
Let
$\Theta_1$
and
$\Theta_2$
be the canonical one-forms on
\[
T^*(\Q\times\Q)
\qquad\text{and}\qquad
T^*(\Q\times\Q^*),
\]
respectively.

Their pullbacks define
\[
\theta_1^{d+}
=
(\alpha_\Q^d)^*\Theta_1
=
-p_{0i}\,dq_0^i
+
p_{1i}\,dq_1^i,
\]
and
\[
\theta_2^{d+}
=
(\beta_\Q^{d+})^*\Theta_2
=
p_{0i}\,dq_0^i
+
q_1^i\,dp_{1i}.
\]

The associated symplectic form on
\[
T^*\Q\times T^*\Q
\]
is
\begin{equation}
\Omega_{T^*\Q\times T^*\Q}
=
-d\theta_1^{d+}
=
d\theta_2^{d+}
=
dq_1^i\wedge dp_{1i}
-
dq_0^i\wedge dp_{0i}
=
-\operatorname{pr}_1^*\omega_\Q+\operatorname{pr}_2^*\omega_\Q ,
\label{eq:disc-Omega}
\end{equation}
that is, the space of one-step relations is $\overline{T^*\Q}\times T^*\Q$, the first
factor carrying the reversed sign. This is the form with respect to which ``Lagrangian
relation'' is to be understood throughout.

As in the continuous case the two wings are not symplectic for the same choice of
canonical form on the targets: with $\omega=-d\Theta$ on both,
$(\alpha_\Q^d)^*(-d\Theta_1)=\Omega_{T^*\Q\times T^*\Q}$ while
$(\beta_\Q^{d+})^*(-d\Theta_2)=-\Omega_{T^*\Q\times T^*\Q}$, so the right wing is
anti-symplectic unless $T^*(\Q\times\Q^*)$ is equipped with the reversed form
$+d\Theta_2$, which is the convention used in the standard discrete Tulczyjew
construction. Nothing below depends on the choice, Lagrangian submanifolds being
insensitive to an overall sign.

\subsubsection{Discrete Lagrangian submanifolds of the discrete Tulczyjew triple}

The symplectic manifold
\[
\bigl(
T^*\Q\times T^*\Q,
\Omega_{T^*\Q\times T^*\Q}
\bigr)
\]
plays the role of the Tulczyjew symplectic space in the discrete theory.

In the framework developed here, implicit discrete Hamiltonian systems are represented by Lagrangian submanifolds
\[
E\subset T^*\Q\times T^*\Q.
\]

Suppose that
\[
E=\operatorname{Im}(X)
\]
for a section $X$ of the mixed projection
\[
\rho:T^*\Q\times T^*\Q\longrightarrow\Q\times\Q^*,
\qquad
\rho\bigl((q_0,p_0),(q_1,p_1)\bigr)=(q_0,p_1),
\]
namely
\[
X(q_0,p_1)
=
\Bigl(
\bigl(q_0,\phi_1(q_0,p_1)\bigr),
\bigl(\phi_0(q_0,p_1),p_1\bigr)
\Bigr).
\]
We emphasise that $X$ is a section of $\rho$ and not a vector field on
$\Q\times\Q^*$; the latter would have its image in $T(\Q\times\Q^*)$ rather than in
$T^*\Q\times T^*\Q$, and the component functions $\phi_0$ and $\phi_1$ carry the values
of $q_1$ and $p_0$ respectively.

Then $E$ is locally described by
\begin{equation}\label{eq:discE}
E=
\left\{
(q_0,p_0,q_1,p_1)
\;\middle|\;
q_1=\phi_0(q_0,p_1),
\qquad
p_0=\phi_1(q_0,p_1)
\right\}.
\end{equation}

If $E$ is Lagrangian, the induced one-form
\[
\phi
=
X^*\theta_2^{d+}
=
\phi_1\,dq_0
+
\phi_0\,dp_1
\]
is closed. Therefore, locally there exists a function
\[
H:\Q\times\Q^*\rightarrow\mathbb R
\]
such that
\[
dH=\phi.
\]

Consequently,
\[
q_1=\frac{\partial H}{\partial p_1},
\qquad
p_0=\frac{\partial H}{\partial q_0}.
\]

Thus explicit discrete Hamiltonian systems appear as particular Lagrangian submanifolds of
\[
T^*\Q\times T^*\Q.
\]

\subsubsection{Discrete Morse families on Tulczyjew's space}
\label{sec:discrete-morse-tulczyjew}

Lagrangian submanifolds of
\[
T^*\Q\times T^*\Q
\]
may be generated locally by discrete Morse families.

Let
\[
F(q_0,p_1,\lambda)
\]
be a discrete Morse family defined on a bundle
\[
\Pp_d
\longrightarrow
\Q\times\Q^*.
\]

The corresponding Lagrangian submanifold is obtained through the right wing
\[
\beta_\Q^{d+}:
T^*\Q\times T^*\Q
\longrightarrow
T^*(\Q\times\Q^*)
\]
of the discrete Tulczyjew triple.

The construction is summarized by
\begin{equation}\label{eq:discMorseGen}
\xymatrix{
\mathbb R
&
\Pp_d \ar[d]^{\pi_d} \ar[l]^{F}
&
T^*(\Q\times\Q^*) \ar[d]_{\pi_{\Q\times\Q^*}}
&
T^*\Q\times T^*\Q
\ar[l]_{\beta_\Q^{d+}}
\ar[dl]
\\
&
\N_d \ar@{^{(}->}[r]
&
\Q\times\Q^* .
}
\end{equation}

In local coordinates, the generated Lagrangian submanifold is
\begin{equation}\label{eq:discMF}
E=
\left\{
(q_0,p_0,q_1,p_1)
\;\middle|\;
p_0=\frac{\partial F}{\partial q_0},
\qquad
q_1=\frac{\partial F}{\partial p_1},
\qquad
\frac{\partial F}{\partial\lambda^a}=0
\right\}.
\end{equation}

When no auxiliary variables are present, the relation is explicitly parametrised by the
mixed variables $(q_0,p_1)$; it need not, however, be explicit as a time-step map, since
recovering $p_1$ from $p_0=\partial F/\partial q_0(q_0,p_1)$ requires a twist condition
such as
\begin{equation}\label{eq:twist}
\det\frac{\partial^2F}{\partial q_0\,\partial p_1}\neq0 .
\end{equation}
When auxiliary variables are present, the stationarity conditions
\[
\frac{\partial F}{\partial\lambda^a}=0
\]
in addition constrain the relation.

As a simple example, take
\[
F(q_0,p_1,\lambda)
=
H(q_0,p_1)
+
\lambda_\alpha\,\Phi^\alpha(q_0,p_1),
\]
with $\rank d\Phi=m_\lambda$, so that $F$ is a Morse family. By \eqref{eq:discMF} it
generates the constrained discrete Hamiltonian system
\[
\Phi^\alpha(q_0,p_1)=0,
\qquad
p_0=\frac{\partial H}{\partial q_0}
+
\lambda_\alpha\frac{\partial\Phi^\alpha}{\partial q_0},
\qquad
q_1=\frac{\partial H}{\partial p_1}
+
\lambda_\alpha\frac{\partial\Phi^\alpha}{\partial p_1}.
\]
The multiplier terms are essential and cannot be dropped: the unshifted equations
$p_0=\partial H/\partial q_0$, $q_1=\partial H/\partial p_1$ together with $m_\lambda$
constraints would cut out a set of dimension $2n-m_\lambda$, which for $m_\lambda>0$
cannot be Lagrangian in the $4n$-dimensional discrete Tulczyjew space.

These constructions provide the geometric foundation for the discrete Dirac systems and implicit Hamiltonian dynamics studied in the following sections.
\section{Discrete Dirac mechanics}
\label{sec:discrete-dirac}

In this section we construct a discrete counterpart of the induced Dirac structure
\[
D_{\Delta_\Q}
\subset
TT^*\Q\oplus T^*T^*\Q.
\]
We begin by introducing discrete analogues of constraint distributions.

\subsection{Discrete constraint distributions}
\label{sec:discrete-constraints}

A natural discrete analogue of a constraint distribution
\[
\Delta_\Q\subset T\Q
\]
is a subset
\[
\Delta_\Q^{d+}\subset\Q\times\Q.
\]

Let \(\Delta_\Q\subset T\Q\) be a distribution of \emph{constant rank}, and let
\[
\Delta_\Q^\circ\subset T^*\Q
\]
be its annihilator.
Assume that \(\Delta_\Q^\circ\) is locally generated by \(m\) independent one-forms
\[
\{\omega^a\}_{a=1}^m,
\qquad
m=\dim\Q-\dim\Delta_\Q.
\]

In local coordinates,
\[
\omega^a
=
A_i^a(q)\,dq^i,
\]
so that, for \(v_q\in T_q\Q\),
\begin{equation}\label{eq:omega-coord}
\omega^a(q,v_q)
=
A_i^a(q)\,v^i.
\end{equation}

Let
\[
\mathcal R:T\Q\rightarrow\Q
\]
be a retraction, defined and inverted on a neighbourhood of the zero section. Recall
that \(\mathcal R\) satisfies
\[
\mathcal R_q(0_q)=q,
\qquad
T_{0_q}\mathcal R_q=\mathrm{id}_{T_q\Q}.
\]

Using the inverse retraction, which is available for \(q_1\) in a neighbourhood of
\(q_0\), we define the discrete constraint functions
\begin{equation}\label{eq:omega_dplus}
\omega^{a}_{d+}(q_0,q_1)
:=
\omega^a
\bigl(
q_0,
\mathcal R_{q_0}^{-1}(q_1)
\bigr).
\end{equation}

The \((+)\)-discrete constraint distribution is then defined by
\begin{equation}\label{eq:DeltaQd}
\Delta_\Q^{d+}
=
\left\{
(q_0,q_1)\in\Q\times\Q
\;\middle|\;
\omega^a_{d+}(q_0,q_1)=0,
\quad
a=1,\dots,m
\right\}.
\end{equation}

The corresponding lifted discrete distribution is
\[
\Delta_{T^*\Q}^{d+}
:=
(\pi_\Q\times\pi_\Q)^{-1}
(\Delta_\Q^{d+})
\subset
T^*\Q\times T^*\Q.
\]

Let
\[
\pi_\Q^{d+}:\Q\times\Q^*\rightarrow\Q,
\qquad
(q_0,p_1)\mapsto q_0.
\]

The annihilator \(\Delta_\Q^\circ\) induces a codistribution
\[
\Delta_{\Q\times\Q^*}^{\circ}
:=
(\pi_\Q^{d+})^*
(\Delta_\Q^\circ)
\subset
T^*(\Q\times\Q^*).
\]

The associated \((+)\)-discrete induced Dirac structure is defined by
\begin{equation}\label{eq:discDiracStructure}
D_{\Delta_\Q}^{d+}
=
\Bigl\{
\bigl((z_0,z_1),\alpha\bigr)
\;\Big|\;
(z_0,z_1)\in\Delta_{T^*\Q}^{d+},
\quad
\alpha-\beta_\Q^{d+}(z_0,z_1)
\in
\Delta_{\Q\times\Q^*}^{\circ}
\Bigr\},
\end{equation}
where
\[
z_0=(q_0,p_0),
\qquad
z_1=(q_1,p_1),
\]
and where the difference is taken at a common base point: \(\alpha\) is required to lie
in \(T^*_{(q_0,p_1)}(\Q\times\Q^*)\), the same fibre as \(\beta_\Q^{d+}(z_0,z_1)\), so
that \(D_{\Delta_\Q}^{d+}\) is a subbundle of the fibre product
\[
(T^*\Q\times T^*\Q)
\times_{\Q\times\Q^*}
T^*(\Q\times\Q^*)
\]
over \(\Q\times\Q^*\).

This construction is the discrete counterpart of the induced Dirac structure appearing
in continuous nonholonomic mechanics. As is customary in nonholonomic mechanics, we use
``Dirac structure'' to mean maximal isotropy with respect to the pairing induced by
\(\beta_\Q^{d+}\), without requiring Courant integrability; for a non-integrable
\(\Delta_\Q\) the object is an almost Dirac structure.

\subsection{$(+)$--Discrete Lagrange--Dirac systems}

Let
\[
L_d:\Q\times\Q\rightarrow\mathbb R
\]
be a discrete Lagrangian.

The \((+)\)-discrete Dirac differential is defined by
\[
\mathfrak D^+L_d
:=
\Gamma_\Q^{d+}\circ dL_d,
\qquad
\Gamma_\Q^{d+}
=
\beta_\Q^{d+}\circ(\alpha_\Q^d)^{-1}.
\]

In local coordinates,
\[
\mathfrak D^+L_d(q_0,q_1)
=
\bigl(
q_0,
D_2L_d(q_0,q_1),
-D_1L_d(q_0,q_1),
q_1
\bigr).
\]

A \((+)\)-discrete Lagrange--Dirac system is specified by the condition
\[
\bigl(
X_d,
\mathfrak D^+L_d(q_0,q_1)
\bigr)
\in
D_{\Delta_\Q}^{d+},
\]
where
\[
X_d=
\bigl((q_0,p_0),(q_1,p_1)\bigr).
\]

This condition yields the discrete Lagrange--Dirac equations
\begin{equation}\label{eq:dLD}
(q_0,q_1)\in\Delta_\Q^{d+},
\qquad
p_1=D_2L_d(q_0,q_1),
\qquad
p_0+D_1L_d(q_0,q_1)
\in
\Delta_\Q^\circ(q_0).
\end{equation}

Equivalently,
\[
\omega_{d+}^a(q_0,q_1)=0,
\qquad
p_1=D_2L_d(q_0,q_1),
\qquad
p_0+D_1L_d(q_0,q_1)
=
\mu_a\,\omega^a(q_0).
\]

\subsection{$(+)$--Discrete nonholonomic Hamiltonian systems}

Let
\[
H_{d+}:\Q\times\Q^*\rightarrow\mathbb R
\]
be a discrete Hamiltonian.

A \((+)\)-discrete nonholonomic Hamiltonian system is defined by
\[
\bigl(
X_d,
dH_{d+}(q_0,p_1)
\bigr)
\in
D_{\Delta_\Q}^{d+}.
\]

This condition yields the equations
\begin{equation}\label{eq:dNH}
(q_0,q_1)\in\Delta_\Q^{d+},
\qquad
q_1=D_2H_{d+}(q_0,p_1),
\qquad
p_0-D_1H_{d+}(q_0,p_1)
\in
\Delta_\Q^\circ(q_0).
\end{equation}

Equivalently,
\[
\omega_{d+}^a(q_0,q_1)=0,
\qquad
q_1=D_2H_{d+}(q_0,p_1),
\qquad
p_0-D_1H_{d+}(q_0,p_1)
=
\mu_a\,\omega^a(q_0).
\]

When
\[
\Delta_\Q=T\Q,
\]
the annihilator vanishes and the above equations reduce to the standard right discrete Hamilton equations.

The discrete Dirac formulation provides a unified framework for regular, degenerate, and
constrained discrete systems. It should be stressed, however, that the unconstrained
relations generated by $dL_d$, by $dH_{d+}$, or more generally by a discrete Morse
family are Lagrangian submanifolds of the discrete Tulczyjew space, whereas a general
nonholonomic Lagrange--Dirac or Hamilton--Dirac update is encoded by the induced discrete
Dirac relation and is \emph{not} Lagrangian in $\overline{T^*\Q}\times T^*\Q$ without
further hypotheses --- this is the discrete counterpart of the classical fact that
nonholonomic flows need not be symplectic. The Hamilton--Jacobi theory of the next
section therefore applies directly to the systems of
Section~\ref{sec:discrete-morse-tulczyjew} and to any nonholonomic system whose update
happens to be presented by a Morse family, but a connecting theorem would be required to
cover every system of the present section.
\section{Discrete Hamilton--Jacobi theory}
\label{sec:discrete-HJ}

The geometric structures developed in the previous sections provide a natural framework for formulating a Hamilton--Jacobi theory for implicit discrete Hamiltonian systems. Since the dynamics is described by Lagrangian submanifolds generated by discrete Morse families, the Hamilton--Jacobi equation should also be formulated intrinsically in terms of these geometric objects, rather than in terms of explicit discrete Hamiltonian maps.

The main idea is to replace the integration of the implicit discrete dynamics by the integration of a reduced discrete system on the configuration manifold. This reduction is achieved by means of a closed one-form whose image defines a Lagrangian submanifold of the discrete phase space. The Hamilton--Jacobi equation is then obtained as the compatibility condition between the projected dynamics on the configuration manifold and the Lagrangian submanifold describing the full discrete dynamics.

A discrete Hamilton--Jacobi theory for explicit discrete Hamiltonian maps was
developed by Ohsawa, Bloch and Leok \citep{OhsawaBlochLeok2011}. The formulation
below differs from theirs in two respects: it is stated in the mixed Type--II
variables \((\bq_k,\bp_{k+1})\), so that the one-forms at consecutive steps are
distinguished rather than identified through a single section; and it applies to
Lagrangian relations generated by Morse families, hence to implicit, constrained
and degenerate discrete systems, and to Lagrangian submanifolds that are not
graphs.

We begin by recalling how Type-I and Type-II Morse families generate Lagrangian submanifolds through the left and right wings of the discrete Tulczyjew triple. We then establish the discrete Hamilton--Jacobi theorem and show how the same geometric construction applies without modification to constrained and degenerate discrete systems.

\subsection{Type-I Morse families}

Type-I Morse families provide the geometric description of implicit discrete dynamics in terms of the left wing of the discrete Tulczyjew triple. They constitute the discrete analogue of the generating families associated with Lagrangian submanifolds of $T^{*}TQ$ in the continuous theory. Let $F(\bq_0,\bq_1,\lambda)$
be a Morse family defined on a fiber bundle $\pi:\Pp_d\longrightarrow \Q\times\Q.$

By \eqref{eq:LF} it generates the immersed Lagrangian submanifold
\[
L_F
=
\left\{
\left(\bq_0,\bq_1;
\frac{\partial F}{\partial\bq_0},
\frac{\partial F}{\partial\bq_1}\right)
\;\middle|\;
\frac{\partial F}{\partial\lambda^a}=0
\right\}
\subset
T^{*}(\Q\times\Q).
\]
We write $L_F$ rather than $\operatorname{Im}(dF)$ because $dF$ takes values in
$T^*\Pp_d$, not in the cotangent bundle of the base; the two agree only when no
auxiliary variables are present.

Using the left wing of the discrete Tulczyjew triple,

\[
\alpha_\Q^d:
T^{*}\Q\times T^{*}\Q
\longrightarrow
T^{*}(\Q\times\Q),
\]

we obtain the Lagrangian submanifold

\[
E
=
(\alpha_\Q^d)^{-1}
(L_F)
\subset
T^{*}\Q\times T^{*}\Q.
\]

In local coordinates,

\[
E
=
\left\{
(\bq_0,\bp_0;\bq_1,\bp_1)
\;\middle|\;
\bp_0
=
-
\frac{\partial F}{\partial\bq_0},
\quad
\bp_1
=
\frac{\partial F}{\partial\bq_1},
\quad
\frac{\partial F}{\partial\lambda^a}
=
0
\right\}.
\]

The stationarity conditions with respect to the auxiliary variables determine the implicit constraints defining the dynamics, while the derivatives with respect to the discrete configuration variables determine the associated momenta. Consequently, every regular Type-I Morse family generates a Lagrangian submanifold of the discrete Tulczyjew space and therefore an implicit discrete Hamiltonian system.

\subsection{Type-II Morse families}

Although both Type-I and Type-II Morse families generate Lagrangian submanifolds of the discrete Tulczyjew space, the latter provide the natural framework for describing implicit discrete Hamiltonian systems. Indeed, they are defined on the mixed variables $(\bq_0,\bp_1),$ which constitute the natural coordinates for discrete Hamiltonian dynamics generated by Type-II generating functions. Let $F(\bq_0,\bp_1,\lambda)$ be a Morse family defined on a fiber bundle
$\pi:\Pp_d\longrightarrow \Q\times\Q^*.$ By \eqref{eq:LF} it generates an immersed Lagrangian submanifold $L_F\subset T^*(\Q\times\Q^*)$, again written $L_F$ rather than $\operatorname{Im}(dF)$ because $dF$ takes values in $T^*\Pp_d$.

Using the right wing of the discrete Tulczyjew triple,

\[
\beta_\Q^{d+}:
T^*\Q\times T^*\Q
\longrightarrow
T^*(\Q\times\Q^*),
\]

we obtain the Lagrangian submanifold

\[
E
=
(\beta_\Q^{d+})^{-1}
(L_F)
\subset
T^*\Q\times T^*\Q.
\]

In local coordinates,

\[
E
=
\left\{
(\bq_0,\bp_0;\bq_1,\bp_1)
\;\middle|\;
\bq_1
=
\frac{\partial F}{\partial\bp_1},
\quad
\bp_0
=
\frac{\partial F}{\partial\bq_0},
\quad
\frac{\partial F}{\partial\lambda^a}
=
0
\right\}.
\]

As in the Type-I case, the auxiliary variables encode the implicit constraints defining the dynamics through the stationarity conditions
\[
\frac{\partial F}{\partial\lambda^a}=0.
\]
The remaining equations determine the discrete Hamiltonian evolution in terms of the mixed variables $(\bq_0,\bp_1).$ Since the Hamilton--Jacobi theory developed below is naturally formulated on the mixed phase space $\Q\times\Q^*,$ Type-II Morse families will be used throughout the remainder of this work.

\subsection{Discrete Hamilton--Jacobi theory}

We now formulate the discrete Hamilton--Jacobi theory in a form adapted to
Type--II discrete Hamiltonian relations. The essential point is that a
Type--II generating family depends on the mixed variables
\[
(\bq_k,\bp_{k+1}),
\]
and therefore the momenta at two consecutive steps must be distinguished.

Let
\[
E_k
=
(\beta_{\mathcal Q}^{d+})^{-1}
(L_{F_k})
\subset
T^*\mathcal Q\times T^*\mathcal Q
\]
be the Lagrangian relation generated by a regular Type--II Morse family
\[
F_k(\bq_k,\bp_{k+1},\lambda).
\]
Its local defining equations are
\begin{equation}
\bp_k
=
\frac{\partial F_k}{\partial\bq_k},
\qquad
\bq_{k+1}
=
\frac{\partial F_k}{\partial\bp_{k+1}},
\qquad
\frac{\partial F_k}{\partial\lambda}
=
0.
\label{eq:typeII-relation-HJ}
\end{equation}

Let
\[
\gamma_k:\mathcal Q\longrightarrow T^*\mathcal Q,
\qquad
\gamma_{k+1}:\mathcal Q\longrightarrow T^*\mathcal Q
\]
be closed one-forms. Their images
\[
\Lambda_k=\operatorname{Im}\gamma_k,
\qquad
\Lambda_{k+1}=\operatorname{Im}\gamma_{k+1}
\]
are Lagrangian submanifolds of \(T^*\mathcal Q\).

The discrete Hamilton--Jacobi problem consists of determining
\(\gamma_k\) and \(\gamma_{k+1}\) so that the Lagrangian relation
\(E_k\) maps \(\Lambda_k\) into \(\Lambda_{k+1}\). More precisely, we
require
\begin{equation}
E_k\circ\Lambda_k
\subset
\Lambda_{k+1}.
\label{eq:HJ-Lagrangian-propagation}
\end{equation}

Suppose that a point
\[
(\bq_k,\bp_k;\bq_{k+1},\bp_{k+1})\in E_k
\]
has both of its momenta on the respective sections,
$\bp_k=\gamma_k(\bq_k)$ and $\bp_{k+1}=\gamma_{k+1}(\bq_{k+1})$. Substituting these two
relations into the defining equations \eqref{eq:typeII-relation-HJ} of $E_k$ gives
\begin{equation}
\gamma_k(\bq_k)
=
\frac{\partial F_k}{\partial\bq_k}
\left(
\bq_k,
\gamma_{k+1}(\bq_{k+1}),
\lambda
\right),
\label{eq:HJ-gamma-k}
\end{equation}
\begin{equation}
\bq_{k+1}
=
\frac{\partial F_k}{\partial\bp_{k+1}}
\left(
\bq_k,
\gamma_{k+1}(\bq_{k+1}),
\lambda
\right),
\label{eq:HJ-q-kplus1}
\end{equation}
together with
\begin{equation}
\frac{\partial F_k}{\partial\lambda}
\left(
\bq_k,
\gamma_{k+1}(\bq_{k+1}),
\lambda
\right)
=
0.
\label{eq:HJ-critical}
\end{equation}

These equations are necessary for such a point, and they express the propagation of one
Lagrangian section into the next under the relation generated by the Type--II Morse
family. They are not, however, equivalent to
\eqref{eq:HJ-Lagrangian-propagation}: substituting
$\bp_{k+1}=\gamma_{k+1}(\bq_{k+1})$ before imposing anything presupposes the conclusion,
so \eqref{eq:HJ-gamma-k}--\eqref{eq:HJ-critical} see only those branches of the critical
set on which it already holds. Wherever $\mathcal C_k$ is non-empty, requiring that some branch satisfy them is
strictly weaker than the inclusion, which asks that \emph{every} branch do so. The next theorem
states the equivalence with the correct quantifier.

\begin{theorem}[Discrete Hamilton--Jacobi theorem]
\label{maintheorem}
Let
\[
E_k
=
(\beta_{\mathcal Q}^{d+})^{-1}
(L_{F_k})
\]
be the Lagrangian relation generated by a regular Type--II Morse family
\[
F_k(\bq_k,\bp_{k+1},\lambda),
\]
and let \(\gamma_k\) and \(\gamma_{k+1}\) be closed one-forms on
\(\mathcal Q\). Denote by
\begin{equation}
\mathcal C_k
=
\left\{
(\bq_k,\bp_{k+1},\lambda)\in\Pp_d
\;\middle|\;
\frac{\partial F_k}{\partial\lambda}=0,
\quad
\gamma_k(\bq_k)
=
\frac{\partial F_k}{\partial\bq_k}
(\bq_k,\bp_{k+1},\lambda)
\right\}
\label{eq:HJ-theorem-1}
\end{equation}
the set of critical points of \(F_k\) whose incoming momentum lies on
\(\operatorname{Im}\gamma_k\).

Then
\[
E_k\circ\operatorname{Im}\gamma_k
\subset
\operatorname{Im}\gamma_{k+1}
\]
if and only if, for \emph{every} \((\bq_k,\bp_{k+1},\lambda)\in\mathcal C_k\),
\begin{equation}
\bp_{k+1}
=
\gamma_{k+1}
\!\left(
\frac{\partial F_k}{\partial\bp_{k+1}}
(\bq_k,\bp_{k+1},\lambda)
\right).
\label{eq:HJ-theorem-2}
\end{equation}
\end{theorem}

\begin{remark}
The quantifier matters. When the critical set
\(\partial F_k/\partial\lambda=0\) has several branches over a given point of
\(\Q\times\Q^*\) --- which is the situation the Morse-family formulation is
designed for --- one branch may satisfy \eqref{eq:HJ-theorem-2} while another
carries the propagated point off \(\operatorname{Im}\gamma_{k+1}\). It is the
condition holding on all of \(\mathcal C_k\), not on some element of it, that is
equivalent to the inclusion. Note also that the propagated base point is now read off
directly from a point of \(\mathcal C_k\) as
\(\bq_{k+1}=\partial F_k/\partial\bp_{k+1}(\bq_k,\bp_{k+1},\lambda)\),
rather than being defined implicitly through itself as in
\eqref{eq:HJ-q-kplus1}.
\end{remark}

\begin{proof}
By definition of the Lagrangian relation generated by the Type--II
Morse family,
\[
(\bq_k,\bp_k;\bq_{k+1},\bp_{k+1})\in E_k
\]
if and only if
\[
\bp_k
=
\frac{\partial F_k}{\partial\bq_k},
\qquad
\bq_{k+1}
=
\frac{\partial F_k}{\partial\bp_{k+1}},
\qquad
\frac{\partial F_k}{\partial\lambda}
=
0,
\]
all three evaluated at \((\bq_k,\bp_{k+1},\lambda)\). Such a point has its
incoming momentum on \(\operatorname{Im}\gamma_k\) precisely when
\(\gamma_k(\bq_k)=\partial F_k/\partial\bq_k\), that is, precisely when
\((\bq_k,\bp_{k+1},\lambda)\in\mathcal C_k\). Hence the composition is the set
\[
E_k\circ\operatorname{Im}\gamma_k
=
\left\{
\left(
\frac{\partial F_k}{\partial\bp_{k+1}}
(\bq_k,\bp_{k+1},\lambda),
\;
\bp_{k+1}
\right)
\;\middle|\;
(\bq_k,\bp_{k+1},\lambda)\in\mathcal C_k
\right\},
\]
parametrised by \(\mathcal C_k\).

An element of this set lies in \(\operatorname{Im}\gamma_{k+1}\) if and only if
its momentum agrees with the value of \(\gamma_{k+1}\) at its base point
\(\bq_{k+1}=\partial F_k/\partial\bp_{k+1}\), which is exactly
\eqref{eq:HJ-theorem-2} for that element. The inclusion
\(E_k\circ\operatorname{Im}\gamma_k\subset\operatorname{Im}\gamma_{k+1}\)
therefore holds if and only if \eqref{eq:HJ-theorem-2} holds at every point of
\(\mathcal C_k\).
\end{proof}

\noindent
\textbf{Exact discrete one-forms.}
Suppose now that
\[
\gamma_k=dW_k,
\qquad
\gamma_{k+1}=dW_{k+1},
\]
for smooth functions
\[
W_k,W_{k+1}:\mathcal Q\longrightarrow\mathbb R.
\]
The set $\mathcal C_k$ of \eqref{eq:HJ-theorem-1} then reads
\[
\mathcal C_k
=
\left\{
(\bq_k,\bp_{k+1},\lambda)
\;\middle|\;
\frac{\partial F_k}{\partial\lambda}=0,
\quad
\frac{\partial W_k}{\partial\bq_k}(\bq_k)
=
\frac{\partial F_k}{\partial\bq_k}(\bq_k,\bp_{k+1},\lambda)
\right\},
\]
and Theorem~\ref{maintheorem} states that
$E_k\circ\operatorname{Im}dW_k\subset\operatorname{Im}dW_{k+1}$ if and only if
\begin{equation}
\bp_{k+1}
=
\frac{\partial W_{k+1}}{\partial\bq_{k+1}}
\!\left(
\frac{\partial F_k}{\partial\bp_{k+1}}
(\bq_k,\bp_{k+1},\lambda)
\right)
\qquad
\text{for every }(\bq_k,\bp_{k+1},\lambda)\in\mathcal C_k .
\label{eq:exact-HJ-1}
\end{equation}
This is the discrete Hamilton--Jacobi equation relating $W_k$ and $W_{k+1}$. Along a
single branch, that is, after the substitution
$\bp_{k+1}=\partial W_{k+1}/\partial\bq_{k+1}$, it takes the more familiar shape
\begin{equation}
\bq_{k+1}
=
\frac{\partial F_k}{\partial\bp_{k+1}}
\left(
\bq_k,
\frac{\partial W_{k+1}}{\partial\bq_{k+1}},
\lambda
\right),
\label{eq:exact-HJ-2}
\end{equation}
together with
\begin{equation}
\frac{\partial F_k}{\partial\lambda}
\left(
\bq_k,
\frac{\partial W_{k+1}}{\partial\bq_{k+1}},
\lambda
\right)
=
0.
\label{eq:exact-HJ-critical}
\end{equation}

As before, this branchwise form is weaker than \eqref{eq:exact-HJ-1} as soon as the
critical set has more than one branch over a given point, and it is
\eqref{eq:exact-HJ-1} that is equivalent to the propagation property.

In contrast with a stationary
formulation, the Type--II structure naturally distinguishes the
incoming momentum
\[
\bp_k=dW_k(\bq_k)
\]
from the outgoing momentum
\[
\bp_{k+1}=dW_{k+1}(\bq_{k+1}).
\]

This distinction becomes essential when the projection of the
Lagrangian submanifold onto configuration space ceases to be regular.
In that case a single function \(W_k\) may no longer describe the
complete Lagrangian submanifold, whereas a Morse family continues to
provide a regular generating description.
\subsection{Propagation of Morse families}

The previous formulation assumes that the Lagrangian submanifold at
each step can locally be represented as the image of a closed one-form.
This representation may fail when the projection onto configuration
space develops singularities. In that situation the natural object is
a Morse family.

Let
\[
S_k(\bq_k,a)
\]
be a Morse family generating a Lagrangian submanifold
\[
\Lambda_k
=
\left\{
\left(
\bq_k,
\frac{\partial S_k}{\partial\bq_k}
\right)
\; ;\;
\frac{\partial S_k}{\partial a}=0
\right\}
\subset T^*\mathcal Q.
\]

Let the discrete dynamics from step \(k\) to \(k+1\) be generated by the
Type--II Morse family
\[
F_k(\bq_k,\bp_{k+1},\lambda).
\]

Define the composition family
\begin{equation}
\mathcal S_{k+1}
(\bq_{k+1};
\bq_k,\bp_{k+1},a,\lambda)
=
S_k(\bq_k,a)
+
\langle\bp_{k+1},\bq_{k+1}\rangle
-
F_k(\bq_k,\bp_{k+1},\lambda).
\label{eq:general-composition-family}
\end{equation}

The variables
\[
(\bq_k,\bp_{k+1},a,\lambda)
\]
are regarded as auxiliary variables, whereas \(\bq_{k+1}\) is the base
variable of the propagated Lagrangian submanifold.

The criticality equations are
\begin{align}
0
&=
\frac{\partial\mathcal S_{k+1}}{\partial\bq_k}
=
\frac{\partial S_k}{\partial\bq_k}
-
\frac{\partial F_k}{\partial\bq_k},
\label{eq:composition-q}
\\
0
&=
\frac{\partial\mathcal S_{k+1}}{\partial\bp_{k+1}}
=
\bq_{k+1}
-
\frac{\partial F_k}{\partial\bp_{k+1}},
\label{eq:composition-p}
\\
0
&=
\frac{\partial\mathcal S_{k+1}}{\partial a}
=
\frac{\partial S_k}{\partial a},
\label{eq:composition-a}
\\
0
&=
\frac{\partial\mathcal S_{k+1}}{\partial\lambda}
=
-
\frac{\partial F_k}{\partial\lambda}.
\label{eq:composition-lambda}
\end{align}

Since
\[
\bp_k
=
\frac{\partial S_k}{\partial\bq_k},
\]
the first two equations recover precisely the Type--II discrete
Hamilton equations
\[
\bp_k
=
\frac{\partial F_k}{\partial\bq_k},
\qquad
\bq_{k+1}
=
\frac{\partial F_k}{\partial\bp_{k+1}},
\]
while the remaining equations retain the critical sets of both Morse
families.

Moreover,
\[
\bp_{k+1}
=
\frac{\partial\mathcal S_{k+1}}
{\partial\bq_{k+1}}.
\]
Consequently, the propagated Lagrangian submanifold is
\begin{equation}
\Lambda_{k+1}
=
\left\{
\left(
\bq_{k+1},
\frac{\partial\mathcal S_{k+1}}
{\partial\bq_{k+1}}
\right)
\; ;\;
D_{(\bq_k,\bp_{k+1},a,\lambda)}
\mathcal S_{k+1}
=
0
\right\}.
\label{eq:propagated-general-Lagrangian}
\end{equation}

This construction propagates the complete Lagrangian submanifold
without requiring it to be represented as the graph of a single
one-form. When the projection onto configuration space is regular,
the auxiliary variables $a$ may locally be eliminated and the usual
Hamilton--Jacobi description is recovered. When the projection becomes
singular, the Morse-family representation remains valid.

Two points must be settled before \eqref{eq:general-composition-family} can be
regarded as defining an integrator. First, \eqref{eq:propagated-general-Lagrangian}
describes a Lagrangian submanifold only if $\mathcal S_{k+1}$ is itself a Morse
family, which does not follow from $S_k$ and $F_k$ being Morse families; this is
settled in Section~\ref{sec:composition-morse}. Second, the composition
introduces $2n$ new auxiliary variables $(\bq_k,\bp_{k+1})$ at every step, so
iterating \eqref{eq:general-composition-family} without further argument produces a
family whose auxiliary dimension grows by $2n+m_\lambda$ at each step;
Section~\ref{sec:reduction} shows that the $2n$ phase-space variables can be eliminated
again at each step, leaving a family with the auxiliary variables of $S_k$ and $F_k$
only, so that the growth is reduced to $m_\lambda$ per step and disappears entirely when
$m_\lambda=0$.

\subsection{Regularity of the composition family}
\label{sec:composition-morse}

We work in local coordinates, writing $n=\dim\mathcal Q$ and letting $m_a$ and
$m_\lambda$ denote the numbers of auxiliary variables of $S_k$ and of $F_k$. Second
derivatives are abbreviated by their subscripts, all evaluated at the point under
consideration:
\[
S_{qq}=\frac{\partial^2S_k}{\partial\bq_k\,\partial\bq_k},
\quad
S_{qa}=\frac{\partial^2S_k}{\partial\bq_k\,\partial a},
\quad
S_{aa}=\frac{\partial^2S_k}{\partial a\,\partial a},
\]
\[
F_{qq}=\frac{\partial^2F_k}{\partial\bq_k\,\partial\bq_k},
\quad
F_{qp}=\frac{\partial^2F_k}{\partial\bq_k\,\partial\bp_{k+1}},
\quad
F_{pp}=\frac{\partial^2F_k}{\partial\bp_{k+1}\,\partial\bp_{k+1}},
\]
\[
F_{q\lambda}=\frac{\partial^2F_k}{\partial\bq_k\,\partial\lambda},
\quad
F_{p\lambda}=\frac{\partial^2F_k}{\partial\bp_{k+1}\,\partial\lambda},
\quad
F_{\lambda\lambda}=\frac{\partial^2F_k}{\partial\lambda\,\partial\lambda},
\]
with $S_{aq}=S_{qa}^{\top}$, $F_{pq}=F_{qp}^{\top}$, $F_{\lambda q}=F_{q\lambda}^{\top}$
and $F_{\lambda p}=F_{p\lambda}^{\top}$. Note that $F_{qp}$ is an $n\times n$ matrix
that need not be symmetric.

\begin{proposition}
\label{prop:composition}
Let $S_k(\bq_k,a)$ be a Morse family generating $\Lambda_k\subset T^*\mathcal Q$ and
let $F_k(\bq_k,\bp_{k+1},\lambda)$ be a Type--II Morse family generating the
Lagrangian relation $E_k$, and let $\mathcal S_{k+1}$ be the composition family
\eqref{eq:general-composition-family}. Then:
\begin{enumerate}
\item[(i)] the set described by \eqref{eq:propagated-general-Lagrangian} is exactly
$E_k\circ\Lambda_k$, irrespective of any regularity assumption;
\item[(ii)] $\mathcal S_{k+1}$ is a Morse family at a critical point if and only if
the $(n+m_a+m_\lambda)\times(2n+m_a+m_\lambda)$ matrix
\begin{equation}
K
=
\begin{pmatrix}
S_{qq}-F_{qq} & -F_{qp} & S_{qa} & -F_{q\lambda}\\
S_{aq} & 0 & S_{aa} & 0\\
-F_{\lambda q} & -F_{\lambda p} & 0 & -F_{\lambda\lambda}
\end{pmatrix},
\label{eq:composition-criterion}
\end{equation}
obtained from the matrix $M$ displayed in the proof below by deleting the
$\bp_{k+1}$ row of blocks and the $\bq_{k+1}$ column of blocks, has full row rank
there;
\item[(iii)] a convenient sufficient condition for (ii) is that
\begin{equation}
F_{qp}\ \text{be invertible}
\qquad\text{and}\qquad
N:=F_{\lambda p}\,F_{qp}^{-1}\,F_{q\lambda}-F_{\lambda\lambda}
\ \text{be invertible}.
\label{eq:composition-hypotheses}
\end{equation}
\end{enumerate}
When (ii) holds, $E_k\circ\Lambda_k$ is an immersed Lagrangian submanifold of
$T^*\mathcal Q$ of dimension $n$, generated by $\mathcal S_{k+1}$.

The hypothesis that $S_k$ be a Morse family is necessary as well as assumed: if
$\bigl(S_{aq}\ \ S_{aa}\bigr)$ fails to have full row rank at a critical point, then
neither does $K$. The conditions \eqref{eq:composition-hypotheses}, by contrast, are
sufficient but not necessary, as the proof makes clear. When $F_k$ carries no
auxiliary variables, $m_\lambda=0$, the second condition in
\eqref{eq:composition-hypotheses} is vacuous.
\end{proposition}

\begin{proof}
(i) On the critical set, \eqref{eq:composition-a} says $\partial S_k/\partial a=0$,
so, writing $\bp_k=\partial S_k/\partial\bq_k$, the pair $(\bq_k,\bp_k)$ lies on
$\Lambda_k$; \eqref{eq:composition-q} then reads $\bp_k=\partial F_k/\partial\bq_k$,
\eqref{eq:composition-p} reads $\bq_{k+1}=\partial F_k/\partial\bp_{k+1}$ and
\eqref{eq:composition-lambda} reads $\partial F_k/\partial\lambda=0$, which together
are exactly the conditions for $\bigl((\bq_k,\bp_k),(\bq_{k+1},\bp_{k+1})\bigr)\in E_k$.
Since $\bp_{k+1}=\partial\mathcal S_{k+1}/\partial\bq_{k+1}$, the point recorded in
\eqref{eq:propagated-general-Lagrangian} is the image of $(\bq_k,\bp_k)$ under $E_k$.
Reading the same equivalences in the opposite order gives the reverse inclusion.

(ii) Write $u=(\bq_k,\bp_{k+1},a,\lambda)$ for the auxiliary variables and
$\bq_{k+1}$ for the base variable. By definition $\mathcal S_{k+1}$ is a Morse
family precisely when $0$ is a regular value of $u\mapsto\partial\mathcal S_{k+1}/\partial u$,
that is, when the $(2n+m_a+m_\lambda)\times(3n+m_a+m_\lambda)$ matrix
\[
M
=
\left(
\frac{\partial^2\mathcal S_{k+1}}{\partial u\,\partial u}
\quad
\frac{\partial^2\mathcal S_{k+1}}{\partial u\,\partial\bq_{k+1}}
\right)
\]
has full row rank. Differentiating
\eqref{eq:composition-q}--\eqref{eq:composition-lambda} and
$\partial\mathcal S_{k+1}/\partial\bq_{k+1}=\bp_{k+1}$, and ordering the columns as
$(\bq_k,\bp_{k+1},a,\lambda,\bq_{k+1})$,
\[
M
=
\begin{pmatrix}
S_{qq}-F_{qq} & -F_{qp} & S_{qa} & -F_{q\lambda} & 0\\
-F_{pq} & -F_{pp} & 0 & -F_{p\lambda} & I_n\\
S_{aq} & 0 & S_{aa} & 0 & 0\\
-F_{\lambda q} & -F_{\lambda p} & 0 & -F_{\lambda\lambda} & 0
\end{pmatrix}.
\]
Let $c=(c_q,c_p,c_a,c_\lambda)$ satisfy $c^{\top}M=0$. The last column of blocks gives
$c_p^{\top}I_n=0$, hence $c_p=0$; and once $c_p=0$ the row of $M$ indexed by
$\bp_{k+1}$ contributes nothing further, while the last column of blocks is satisfied
identically. The condition $c^{\top}M=0$ is therefore equivalent to
$(c_q,c_a,c_\lambda)^{\top}K=0$ for the matrix $K$ of
\eqref{eq:composition-criterion}. Since $M$ has $2n+m_a+m_\lambda$ rows, $K$ has
$n+m_a+m_\lambda$ rows, and the two left kernels are isomorphic, one has full row rank
if and only if the other does. This proves (ii). Written out, the four columns of $K$
give
\begin{align}
c_q^{\top}(S_{qq}-F_{qq})+c_a^{\top}S_{aq}-c_\lambda^{\top}F_{\lambda q}&=0,
\label{eq:rank-1}\\
c_q^{\top}F_{qp}+c_\lambda^{\top}F_{\lambda p}&=0,
\label{eq:rank-2}\\
c_q^{\top}S_{qa}+c_a^{\top}S_{aa}&=0,
\label{eq:rank-3}\\
c_q^{\top}F_{q\lambda}+c_\lambda^{\top}F_{\lambda\lambda}&=0.
\label{eq:rank-4}
\end{align}

(iii) Assume \eqref{eq:composition-hypotheses}. Because $F_{qp}$ is invertible,
\eqref{eq:rank-2} gives $c_q^{\top}=-c_\lambda^{\top}F_{\lambda p}F_{qp}^{-1}$, and
substituting this into \eqref{eq:rank-4} yields
\[
c_\lambda^{\top}
\bigl(F_{\lambda p}F_{qp}^{-1}F_{q\lambda}-F_{\lambda\lambda}\bigr)
=
c_\lambda^{\top}N
=
0 ,
\]
so $c_\lambda=0$ because $N$ is invertible, and therefore $c_q=0$ as well. With
$c_q=0$, equations \eqref{eq:rank-1} and \eqref{eq:rank-3} become
$c_a^{\top}S_{aq}=0$ and $c_a^{\top}S_{aa}=0$, that is
$c_a^{\top}\bigl(S_{aq}\ \ S_{aa}\bigr)=0$; since $S_k$ is a Morse family this matrix
has full row rank $m_a$, whence $c_a=0$. Thus $c=0$ and $K$ has full row rank.

Whenever (ii) holds, the critical set of $\mathcal S_{k+1}$ is a submanifold of
dimension $n$, and by the construction of
Section~\ref{subsec:morse-families-continuous} the family generates an immersed
Lagrangian submanifold of $T^*\mathcal Q$, which by (i) is $E_k\circ\Lambda_k$.

For the two remaining assertions, note first that if $c_a\neq0$ satisfies
$c_a^{\top}\bigl(S_{aq}\ \ S_{aa}\bigr)=0$ then $(c_q,c_a,c_\lambda)=(0,c_a,0)$ solves
\eqref{eq:rank-1}--\eqref{eq:rank-4}, so $K$ is row rank deficient and $S_k$ being a
Morse family is indeed necessary. On the other hand \eqref{eq:composition-hypotheses}
is not: if $F_{qp}=0$ and $m_\lambda=0$, then \eqref{eq:rank-2} is vacuous and
\eqref{eq:rank-1} together with \eqref{eq:rank-3} becomes the square symmetric system
\[
(c_q^{\top}\ \ c_a^{\top})
\begin{pmatrix}
S_{qq}-F_{qq} & S_{qa}\\
S_{aq} & S_{aa}
\end{pmatrix}
=0 ,
\]
which for generic data has only the solution $c_q=0$, $c_a=0$; in that case $K$ has
full row rank although $F_{qp}$ is singular. Neither is the invertibility of $N$
necessary; \eqref{eq:composition-hypotheses} is a convenient sufficient condition
rather than a characterisation, and \eqref{eq:composition-criterion} is the sharp
statement.
\end{proof}

\begin{remark}
\label{rem:second-class}
The matrix $N$ takes a concrete form in the constrained case. For
\[
F_k(\bq_k,\bp_{k+1},\lambda)
=
H_{d+}(\bq_k,\bp_{k+1})+\lambda_\alpha\Phi^\alpha(\bq_k,\bp_{k+1}),
\]
one has $F_{\lambda\lambda}=0$, so
\[
N
=
\frac{\partial\Phi}{\partial\bp_{k+1}}
\,F_{qp}^{-1}\,
\left(\frac{\partial\Phi}{\partial\bq_k}\right)^{\!\top} ,
\]
an $m_\lambda\times m_\lambda$ matrix. Its invertibility is exactly the condition that
the criticality equations determine the multipliers $\lambda$ uniquely in terms of the
remaining variables; we emphasise that $N$ is not in general antisymmetric and is not
the Dirac matrix $\Phi_q\Phi_p^{\top}-\Phi_p\Phi_q^{\top}$ of the constraint algebra,
so no dictionary with the first- or second-class classification should be read into
it. The condition is insensitive to how the constraints are normalised: replacing
$\Phi$ by $D\Phi$ for an invertible constant matrix $D$, and correspondingly
$\lambda$ by $D^{-\top}\lambda$, leaves $F_k$ unchanged as a function and replaces $N$
by $DND^{\top}$.

The condition on $N$ cannot be dropped from \eqref{eq:composition-hypotheses}, even
when $F_{qp}$ is invertible and both $S_k$ and $F_k$ are Morse families. Taking
$n=m_a=m_\lambda=1$ with
\[
S_k(q,a)=\tfrac12(q+a)^2,
\qquad
F_k(q,p,\lambda)=qp+\lambda p,
\]
both $S_k$ and $F_k$ are Morse families and $F_{qp}=1$ is invertible, but
$N=F_{\lambda p}F_{qp}^{-1}F_{q\lambda}-F_{\lambda\lambda}=0$; the matrix $M$ then has
rank $3$ rather than $4$, its left kernel being spanned by
$(c_q,c_p,c_a,c_\lambda)=(-1,0,1,1)$, and $\mathcal S_{k+1}$ is not a Morse family.
The mechanism is that $\Phi=p$ does not depend on $\bq_k$, so
$\partial\Phi/\partial\bq_k=0$ and $N$ vanishes identically. Note that what fails here
is only the Morse property of the composition family: by
Proposition~\ref{prop:composition}(i) the set it describes is still $E_k\circ\Lambda_k$,
which in this example is the zero section and is perfectly regular. Failure of
\eqref{eq:composition-criterion} obstructs the \emph{generating-family} description of
$E_k\circ\Lambda_k$, not the composition itself.
\end{remark}

\subsection{Reduction of the auxiliary variables}
\label{sec:reduction}

Proposition~\ref{prop:composition} makes $\mathcal S_{k+1}$ a legitimate generating
family, but it is a family in $2n+m_a+m_\lambda$ auxiliary variables, whereas $S_k$
had only $m_a$. Iterating would therefore produce, after $N$ steps, a family with
$2nN+m_a+Nm_\lambda$ auxiliary variables. We now show that the $2n$ variables
introduced by the composition can be removed again at each step, so that the
$2n$ internal phase-space variables per step do not accumulate. Note that the
reduction eliminates only those $2n$ variables: the multipliers $\lambda$ of the step
survive alongside the incoming $a$, so one step turns a family with $m_a$ auxiliary
variables into one with $m_a+m_\lambda$, and $N$ steps give $m_a+Nm_\lambda$. The auxiliary count is genuinely fixed only when $m_\lambda=0$, which
is the case in the application of Section~\ref{sec:optical-application}; in general a
separate argument would be needed to eliminate the newly introduced multipliers as
well.

Assume that $F_k$ is consistent with a continuous Hamiltonian in the sense that
\begin{equation}
F_k(\bq_k,\bp_{k+1},\lambda)
=
\langle\bp_{k+1},\bq_k\rangle
+
h\,G_k(\bq_k,\bp_{k+1},\lambda),
\label{eq:consistent-family}
\end{equation}
which is the form taken by every Type--II discrete Hamiltonian that reduces to the
identity relation as $h\to0$; for the first-order choice one has $G_k=H(z_k,\cdot,\cdot)$.
The constrained families of Remark~\ref{rem:second-class} are of this form after the
multipliers are rescaled: writing $\lambda=h\mu$ turns
$H_{d+}+\lambda_\alpha\Phi^\alpha$ into
$\langle\bp_{k+1},\bq_k\rangle+h(H+\mu_\alpha\Phi^\alpha)$, and since the rescaling
replaces $N$ by $h^2N$ it does not affect the hypotheses of
Proposition~\ref{prop:composition}.
Write
\begin{equation}
B
=
\begin{pmatrix}
\dfrac{\partial^2\mathcal S_{k+1}}{\partial\bq_k\,\partial\bq_k}
&
\dfrac{\partial^2\mathcal S_{k+1}}{\partial\bq_k\,\partial\bp_{k+1}}
\\[2mm]
\dfrac{\partial^2\mathcal S_{k+1}}{\partial\bp_{k+1}\,\partial\bq_k}
&
\dfrac{\partial^2\mathcal S_{k+1}}{\partial\bp_{k+1}\,\partial\bp_{k+1}}
\end{pmatrix}
=
\begin{pmatrix}
S_{qq}-hG_{qq} & -(I_n+hG_{qp})\\
-(I_n+hG_{pq}) & -hG_{pp}
\end{pmatrix}.
\label{eq:B-block}
\end{equation}

\begin{proposition}
\label{prop:reduction}
Suppose that, on the set of arguments $(\bq_k,\bp_{k+1},a,\lambda)$ under
consideration,
\[
\left\|
D^2_{(\bq_k,\bp_{k+1})}G_k
\right\|
\le C_G
\qquad\text{and}\qquad
\left\|S_{qq}\right\|\le C_S ,
\]
and set $h_0=\bigl[C_G(1+C_S)\bigr]^{-1}$. Then for every $0<h<h_0$:
\begin{enumerate}
\item[(i)] $B$ is invertible there, with
$\det B=(-1)^n+\mathcal O(h)$ and
$\|B^{-1}\|\le(1+C_S)\bigl[1-hC_G(1+C_S)\bigr]^{-1}$;
\item[(ii)] near each critical point of $\mathcal S_{k+1}$ the equations
$\partial\mathcal S_{k+1}/\partial\bq_k=0$ and
$\partial\mathcal S_{k+1}/\partial\bp_{k+1}=0$ determine
$(\bq_k,\bp_{k+1})=\bigl(\hat\bq_k,\hat\bp_{k+1}\bigr)(\bq_{k+1},a,\lambda)$
uniquely and smoothly, locally in $(\bq_{k+1},a,\lambda)$; globally single-valued
branches require additional injectivity and domain hypotheses;
\item[(iii)] the reduced function
\begin{equation}
\widetilde S_{k+1}(\bq_{k+1},a,\lambda)
:=
\mathcal S_{k+1}
\bigl(\bq_{k+1};\hat\bq_k,\hat\bp_{k+1},a,\lambda\bigr)
\label{eq:reduced-family}
\end{equation}
satisfies
\[
\frac{\partial\widetilde S_{k+1}}{\partial\bq_{k+1}}=\bp_{k+1},
\qquad
\frac{\partial\widetilde S_{k+1}}{\partial a}=\frac{\partial S_k}{\partial a},
\qquad
\frac{\partial\widetilde S_{k+1}}{\partial\lambda}=-\frac{\partial F_k}{\partial\lambda},
\]
and locally describes the same subset $E_k\circ\Lambda_k$ of $T^*\mathcal Q$ as
$\mathcal S_{k+1}$ does, using only the auxiliary variables $(a,\lambda)$. If in
addition the criterion of Proposition~\ref{prop:composition}(ii) holds, then by (iv)
both are Morse families and both generate the immersed Lagrangian submanifold
$\Lambda_{k+1}=E_k\circ\Lambda_k$; note that the hypotheses of the present proposition
do not by themselves imply that criterion, and if it fails neither function generates a
Lagrangian submanifold;
\item[(iv)] $\widetilde S_{k+1}$ is a Morse family if and only if $\mathcal S_{k+1}$ is.
\end{enumerate}
\end{proposition}

\begin{proof}
(i) At $h=0$ the block \eqref{eq:B-block} becomes
\[
B_0
=
\begin{pmatrix}
S_{qq} & -I_n\\
-I_n & 0
\end{pmatrix},
\qquad
B_0^{-1}
=
\begin{pmatrix}
0 & -I_n\\
-I_n & -S_{qq}
\end{pmatrix},
\]
as one checks by multiplying out; in particular $\det B_0=(-1)^n$ \emph{independently
of $S_{qq}$}, and $\|B_0^{-1}\|\le1+C_S$. Since
$B-B_0=-h\,D^2_{(\bq_k,\bp_{k+1})}G_k$ has norm at most $hC_G$, the Neumann series for
$B^{-1}$ converges as soon as $hC_G(1+C_S)<1$, which is $h<h_0$, and gives the
stated bound on $\|B^{-1}\|$. Continuity of the determinant gives
$\det B=(-1)^n+\mathcal O(h)$.

(ii) Immediate from (i) and the implicit function theorem, $B$ being the Jacobian of
the two equations with respect to $(\bq_k,\bp_{k+1})$.

(iii) Abbreviate $w=(\bq_k,\bp_{k+1})$ and $v=(a,\lambda)$, so that
$\widetilde S_{k+1}(\bq_{k+1},v)=\mathcal S_{k+1}(\bq_{k+1};\hat w(\bq_{k+1},v),v)$
with $\partial_w\mathcal S_{k+1}(\hat w)=0$. Differentiating and using this
criticality, the terms in $\partial\hat w$ drop out, so
$\partial_{\bq_{k+1}}\widetilde S_{k+1}=\partial_{\bq_{k+1}}\mathcal S_{k+1}=\bp_{k+1}$
and $\partial_v\widetilde S_{k+1}=\partial_v\mathcal S_{k+1}$, which is the second and
third identity by \eqref{eq:composition-a} and \eqref{eq:composition-lambda}. The
critical set of $\widetilde S_{k+1}$ is therefore the image of that of
$\mathcal S_{k+1}$ under the projection forgetting $w$, and the two families assign the
same momentum to the same base point, so they generate the same subset of
$T^*\mathcal Q$, namely $E_k\circ\Lambda_k$ by Proposition~\ref{prop:composition}(i).

(iv) Differentiating $\partial_w\mathcal S_{k+1}(\hat w)=0$ gives
$\partial_v\hat w=-B^{-1}\partial^2_{wv}\mathcal S_{k+1}$ and
$\partial_{\bq_{k+1}}\hat w=-B^{-1}\partial^2_{w\bq_{k+1}}\mathcal S_{k+1}$, whence
\[
\left(
\frac{\partial^2\widetilde S_{k+1}}{\partial v\,\partial v}
\quad
\frac{\partial^2\widetilde S_{k+1}}{\partial v\,\partial\bq_{k+1}}
\right)
=
\left(
\frac{\partial^2\mathcal S_{k+1}}{\partial v\,\partial v}
\quad
\frac{\partial^2\mathcal S_{k+1}}{\partial v\,\partial\bq_{k+1}}
\right)
-
\frac{\partial^2\mathcal S_{k+1}}{\partial v\,\partial w}
\,B^{-1}
\left(
\frac{\partial^2\mathcal S_{k+1}}{\partial w\,\partial v}
\quad
\frac{\partial^2\mathcal S_{k+1}}{\partial w\,\partial\bq_{k+1}}
\right).
\]
The right-hand side is exactly what block Gaussian elimination of the $w$-rows
produces from the last $m_a+m_\lambda$ rows of the matrix $M$ of
Proposition~\ref{prop:composition}. Since $B$ is invertible, right multiplication of
the resulting block-triangular matrix by
$\left(\begin{smallmatrix}I&-B^{-1}X\\0&I\end{smallmatrix}\right)$, where $X$ collects
its upper right blocks, clears them without changing the rank, so
\[
\operatorname{rank}M
=
2n
+
\operatorname{rank}
\left(
\frac{\partial^2\widetilde S_{k+1}}{\partial v\,\partial v}
\quad
\frac{\partial^2\widetilde S_{k+1}}{\partial v\,\partial\bq_{k+1}}
\right),
\]
and $M$ has full row rank $2n+m_a+m_\lambda$ if and only if the reduced matrix has
full row rank $m_a+m_\lambda$.
\end{proof}

\begin{corollary}
\label{cor:semidiscrete-HJ}
Under the hypotheses of Proposition~\ref{prop:reduction} with $m_\lambda=0$ and
$G_k=H(z_k,\cdot,\cdot)$,
\begin{equation}
\widetilde S_{k+1}(\bq,a)
=
S_k(\bq,a)
-
h\,H\!\left(z_k,\bq,\frac{\partial S_k}{\partial\bq}(\bq,a)\right)
+
\mathcal O(h^2).
\label{eq:semidiscrete-HJ}
\end{equation}
\end{corollary}

\begin{proof}
Write $\bp^0=\partial S_k(\bq,a)/\partial\bq$. At $h=0$ the criticality equations give
$\hat\bq_k=\bq$ and $\hat\bp_{k+1}=\bp^0$; to first order,
$\bq_{k+1}-\hat\bq_k=h\,\partial H/\partial\bp(z_k,\bq,\bp^0)+\mathcal O(h^2)$ from
\eqref{eq:composition-p}. Substituting into
$\mathcal S_{k+1}=S_k(\hat\bq_k,a)+\langle\hat\bp_{k+1},\bq_{k+1}-\hat\bq_k\rangle
-hH(z_k,\hat\bq_k,\hat\bp_{k+1})$ and expanding
$S_k(\hat\bq_k,a)=S_k(\bq,a)-\langle\bp^0,\bq_{k+1}-\hat\bq_k\rangle+\mathcal O(h^2)$,
the two inner products cancel and \eqref{eq:semidiscrete-HJ} follows.
\end{proof}

Equation \eqref{eq:semidiscrete-HJ} is an explicit Euler discretisation, in the
propagation variable, of the Hamilton--Jacobi equation
\[
\frac{\partial S}{\partial z}
+
H\!\left(z,\bq,\frac{\partial S}{\partial\bq}\right)
=
0
\]
for the \emph{family} $S(\bq,a)$, with the auxiliary variable $a$ carried along as a
passive parameter. The content of Propositions~\ref{prop:composition} and
\ref{prop:reduction} is that this semi-discretisation remains valid where the
classical Hamilton--Jacobi description does not, namely where the projection of
$\Lambda_k$ onto $\mathcal Q$ is singular.

\begin{remark}
\label{rem:through-caustics}
It is worth being precise about why the caustic does not obstruct
Proposition~\ref{prop:reduction}. The hypothesis is a bound on
$S_{qq}=\partial^2S_k/\partial\bq_k\partial\bq_k$, the second derivative of the family
at \emph{fixed} $a$, and not on the curvature of the projected wavefront. These are
different quantities, and only the second one is singular at a caustic: for the fold
family \eqref{eq:optical-Sk-app},
\[
\frac{\partial^2S_k}{\partial x_k^2}
=
S_{\mathrm{reg},k}''(x_k)-\rho_k''(x_k)(a-a_c)
\]
remains bounded on compact sets, whereas the branchwise curvature
$dp_k^{\pm}/dx_k$ obtained from \eqref{eq:optical-two-momenta-app} diverges like
$\rho_k(x_k)^{-1/2}$ as the caustic is approached. Since $\det B_0=(-1)^n$
independently of $S_{qq}$, the wavefront curvature enters the threshold
$h_0=[C_G(1+C_S)]^{-1}$ only through the product $hC_GC_S$, and a fold of the
projection, at which $C_S$ stays bounded, does not by itself obstruct the
reduction. It should be said that $C_S$ is not the only way the reduction can fail:
a stiff Hamiltonian alone will do it, as the scalar example
$S_k(q,a)=qa$, $G_k=\tfrac12 c\,(q^2+p^2)$ shows, for which $C_S=0$ but
$\det B=c^2h^2-1$ vanishes at $h=1/c=h_0$. What is true is that within the family the
reduction survives the caustic, and fails only when $hC_G(1+C_S)$ ceases to be
small --- in particular when the \emph{family} develops a singularity in the base
direction, which makes $C_S$ blow up. At such a point one additional auxiliary variable must be introduced; in
one configuration dimension a single auxiliary variable already generates every
singularity of type $A_k$, through
\[
S(x,\sigma)
=
\frac{\sigma^{k+1}}{k+1}
+
\sum_{j=1}^{k-1}\rho_j(x)\,\sigma^{j},
\]
so the auxiliary dimension never needs to exceed one there.
\end{remark}
\section{Application: discrete propagation of an optical caustic}
\label{sec:optical-application}

We now apply the discrete geometric framework developed above to the
propagation of an optical wavefront through a fold caustic. The essential
feature of this example is that the discrete Hamiltonian dynamics may
remain regular while the projection of the propagated Lagrangian
submanifold onto configuration space becomes singular. Consequently, the
Morse-family description is required for the propagated wavefront rather
than for the discrete Hamiltonian map itself.

This distinction allows us to propagate the complete Lagrangian
submanifold without introducing a separate generating function for every
branch of the multivalued Hamilton--Jacobi solution.

\subsection{Optical Hamiltonian and discrete dynamics}

Consider a two-dimensional optical medium with refractive index
\(n=n(z,x)\), where \(z\) denotes the propagation variable and \(x\) is the
transverse configuration variable. A standard Hamiltonian formulation of
geometrical optics is determined by
\begin{equation}
    H(z,x,p)
    =
    -\sqrt{n^2(z,x)-p^2},
    \label{eq:optical-Hamiltonian-app}
\end{equation}
with Hamilton equations
\begin{equation}
    \frac{dx}{dz}
    =
    \frac{\partial H}{\partial p},
    \qquad
    \frac{dp}{dz}
    =
    -
    \frac{\partial H}{\partial x}.
    \label{eq:optical-Hamilton-equations-app}
\end{equation}
Explicitly,
\begin{equation}
    \frac{dx}{dz}
    =
    \frac{p}{\sqrt{n^2(z,x)-p^2}},
    \qquad
    \frac{dp}{dz}
    =
    \frac{n(z,x)n_x(z,x)}
    {\sqrt{n^2(z,x)-p^2}}.
    \label{eq:optical-Hamilton-equations-explicit-app}
\end{equation}
The Hamiltonian formulation of ray optics and the relation between
Hamiltonian rays and caustics are classical; see, for example,
\cite{KravtsovOrlov1990,KravtsovOrlov1993}.

Throughout this section $n$ is assumed positive and at least twice continuously
differentiable, and all statements are made on a compact subset of the domain
\begin{equation}\label{eq:optical-domain}
\mathcal D
=
\bigl\{(z,x,p)\;:\;|p|<n(z,x)\bigr\},
\end{equation}
on which \eqref{eq:optical-Hamiltonian-app} is defined and smooth. The restriction is
not cosmetic: $|p|=n$ is grazing incidence, where the right-hand sides of
\eqref{eq:optical-Hamilton-equations-explicit-app} blow up and rays turn.

We discretize the propagation variable according to
\[
    z_k=z_0+kh,
\]
where \(h>0\) is the propagation step. Following the Type--II discrete
Hamiltonian construction, we use the first-order approximation
\begin{equation}
    H_{d+}^{\,k}(x_k,p_{k+1})
    =
    x_kp_{k+1}
    +
    hH(z_k,x_k,p_{k+1})
    \label{eq:optical-Hdplus-app}
\end{equation}
to the corresponding exact right discrete Hamiltonian. For
\eqref{eq:optical-Hamiltonian-app}, this gives
\begin{equation}
    H_{d+}^{\,k}(x_k,p_{k+1})
    =
    x_kp_{k+1}
    -
    h\sqrt{
        n^2(z_k,x_k)-p_{k+1}^2
    }.
    \label{eq:optical-Hdplus-explicit-app}
\end{equation}

The associated Type--II discrete Hamilton equations are
\begin{equation}
    p_k
    =
    \frac{\partial H_{d+}^{\,k}}{\partial x_k},
    \qquad
    x_{k+1}
    =
    \frac{\partial H_{d+}^{\,k}}{\partial p_{k+1}}.
    \label{eq:optical-TypeII-app}
\end{equation}
Differentiating \eqref{eq:optical-Hdplus-explicit-app}, we obtain
\begin{equation}
    p_k
    =
    p_{k+1}
    -
    h
    \frac{
        n(z_k,x_k)n_x(z_k,x_k)
    }{
        \sqrt{n^2(z_k,x_k)-p_{k+1}^2}
    },
    \label{eq:optical-discrete-p-app}
\end{equation}
and
\begin{equation}
    x_{k+1}
    =
    x_k
    +
    h
    \frac{
        p_{k+1}
    }{
        \sqrt{n^2(z_k,x_k)-p_{k+1}^2}
    }.
    \label{eq:optical-discrete-x-app}
\end{equation}

Equivalently,
\begin{equation}
    \frac{x_{k+1}-x_k}{h}
    =
    H_p(z_k,x_k,p_{k+1}),
    \qquad
    \frac{p_{k+1}-p_k}{h}
    =
    -H_x(z_k,x_k,p_{k+1}),
    \label{eq:optical-consistency-app}
\end{equation}
so that the scheme is consistent with the continuous optical Hamilton
equations. Notice that no auxiliary variable is required to define this
discrete Hamiltonian relation.

For given \((x_k,p_k)\), the first equation
\eqref{eq:optical-discrete-p-app} is solved implicitly for \(p_{k+1}\),
after which \eqref{eq:optical-discrete-x-app} determines \(x_{k+1}\).
Solvability is local and requires the twist condition \eqref{eq:twist}, which here reads
\begin{equation}\label{eq:optical-twist}
1
-
h\,\frac{p_{k+1}\,n(z_k,x_k)\,n_x(z_k,x_k)}
        {\bigl(n^2(z_k,x_k)-p_{k+1}^2\bigr)^{3/2}}
\neq0 ,
\end{equation}
and which holds on any compact subset of $\mathcal D$ for $h$ small enough. Under
\eqref{eq:optical-twist} the relation is locally a symplectic one-step map,
\begin{equation}
    (x_k,p_k)
    \longmapsto
    (x_{k+1},p_{k+1}),
    \label{eq:optical-integrator-app}
\end{equation}
which we refer to below as the Type--II ray integrator
\eqref{eq:optical-integrator-app}.

\subsection{The propagated wavefront and formation of a fold}

Let \(a\in I\) parameterize the initial family of rays. Thus
\[
    x_0=x_0(a),
    \qquad
    p_0=p_0(a),
\]
and repeated application of the discrete Hamiltonian map gives
\[
    x_k=x_k(a),
    \qquad
    p_k=p_k(a).
\]
At every propagation step, the family determines the Lagrangian
submanifold
\begin{equation}
    \Lambda_k^d
    =
    \left\{
        \bigl(x_k(a),p_k(a)\bigr)
        \; ;\;
        a\in I
    \right\}
    \subset T^*\Q.
    \label{eq:optical-Lambda-k-app}
\end{equation}

The variables \(z\) and \(a\) play different roles. The variable \(z\)
describes propagation along each ray, whereas \(a\) labels the different
rays in the wavefront. At \(z=0\), the points \(x_0(a)\) form the initial
wavefront. Figure~\ref{fig:optical-rays-fold} shows a computed example in which
\(\Lambda_k^d\) remains a smooth curve while its projection develops an envelope.

\begin{figure}[ht]
\centering
\includegraphics[width=0.92\textwidth]{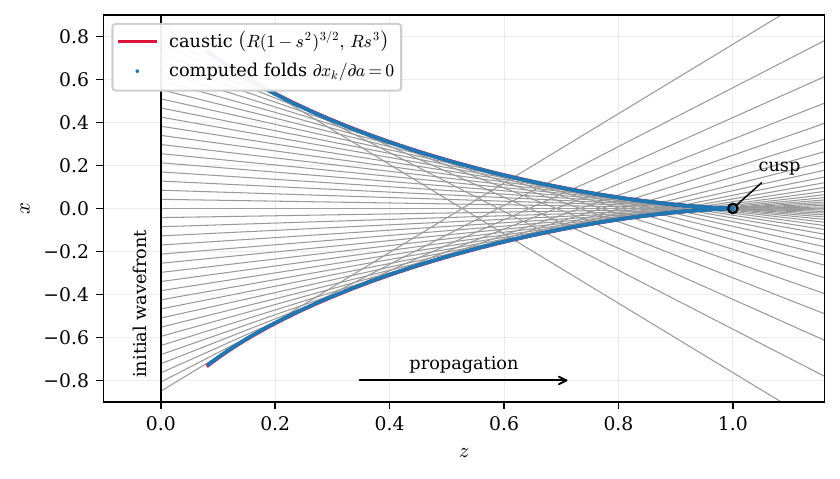}
\caption{Propagation of a one-parameter family of optical rays, computed with the
Type--II scheme of Section~\ref{sec:optical-application} in a homogeneous medium
\(n\equiv1\) from the parabolic initial wavefront \(W_0(x)=-x^2/2R\), \(R=1\).
The variable \(z\) is the propagation variable, whereas \(a\in I\) labels the
rays through their initial positions \(x_0(a)\). A caustic occurs where the
projection of the propagated Lagrangian submanifold onto configuration space
becomes singular; it is the \emph{envelope} of the ray family rather than a point
of common intersection. For this initial wavefront the envelope is available in
closed form as
\(\bigl(z_c(s),x_c(s)\bigr)=\bigl(R(1-s^2)^{3/2},\,Rs^{3}\bigr)\) with
\(s=a/R\) and \(|s|<1\), equivalently as the astroid arc
\(|x_c/R|^{2/3}+(z_c/R)^{2/3}=1\); it consists of two fold branches meeting at a
semicubical cusp at \((R,0)\), near which
\(x_c^2\sim\tfrac{8}{27}(R-z_c)^3/R\). The generic fold condition
\eqref{eq:optical-fold-condition-app} holds along each branch away from the cusp.
The computed points at which \(\partial x_k/\partial a\) changes sign are marked
and lie on this curve.}
\label{fig:optical-rays-fold}
\end{figure}

Consider the projection
\[
    \pi_\Q:\Lambda_k^d\longrightarrow \Q,
    \qquad
    (x_k(a),p_k(a))
    \longmapsto x_k(a).
\]
We assume that the propagated family is an immersed curve, that is, that
\((\partial x_k/\partial a,\partial p_k/\partial a)\neq(0,0)\) for every \(a\); this is
what makes \(\Lambda_k^d\) an immersed one-dimensional submanifold and is not implied by
the fold conditions below. A generic fold of the projection then occurs at \(a=a_c\) when
\begin{equation}
    \frac{\partial x_k}{\partial a}(a_c)=0,
    \qquad
    \frac{\partial^2x_k}{\partial a^2}(a_c)\neq0.
    \label{eq:optical-fold-condition-app}
\end{equation}
At such a point the immersion hypothesis forces \(\partial p_k/\partial a(a_c)\neq0\).
Locally, after a fibre-preserving local reparametrisation of the auxiliary variable ---
which need not agree with the original ray label \(a\) --- and a non-zero rescaling, the
fold can be written in the normal form
\begin{equation}
    (a-a_c)^2=\rho_k(x_k),
    \label{eq:optical-fold-relation-app}
\end{equation}
where
\begin{equation}
    \rho_k(x_c)=0,
    \qquad
    \rho_k'(x_c)\neq0.
    \label{eq:optical-rho-app}
\end{equation}

For \(\rho_k(x_k)>0\), the same configuration point \(x_k\) corresponds to
two values of the ray parameter,
\begin{equation}
    a_\pm(x_k)
    =
    a_c\pm\sqrt{\rho_k(x_k)}.
    \label{eq:optical-two-rays-app}
\end{equation}
Thus the projection is no longer one-to-one, even though the Lagrangian
submanifold in phase space remains regular.

\subsection{One Morse family instead of separate branches}

Near the fold, represent the Lagrangian submanifold by the Morse family
\begin{equation}
    S_k(x_k,a)
    =
    S_{\mathrm{reg},k}(x_k)
    +
    \frac{(a-a_c)^3}{3}
    -
    \rho_k(x_k)(a-a_c).
    \label{eq:optical-Sk-app}
\end{equation}
This is the standard generating-family normal form associated with a fold
singularity.

Its critical set is determined by
\begin{equation}
    \frac{\partial S_k}{\partial a}
    =
    (a-a_c)^2-\rho_k(x_k)
    =
    0.
    \label{eq:optical-Sk-critical-app}
\end{equation}
The Lagrangian submanifold generated by \(S_k\) is therefore
\begin{equation}
    \Lambda_k^d
    =
    \left\{
        \left(
            x_k,
            \frac{\partial S_k}{\partial x_k}(x_k,a)
        \right)
        \; ;\;
        \frac{\partial S_k}{\partial a}(x_k,a)=0
    \right\}.
    \label{eq:optical-Lambda-Morse-app}
\end{equation}

Since
\begin{equation}
    \frac{\partial S_k}{\partial x_k}
    =
    S_{\mathrm{reg},k}'(x_k)
    -
    \rho_k'(x_k)(a-a_c),
    \label{eq:optical-Sx-app}
\end{equation}
the two critical points \eqref{eq:optical-two-rays-app} give
\begin{equation}
    p_k^\pm(x_k)
    =
    S_{\mathrm{reg},k}'(x_k)
    \mp
    \rho_k'(x_k)\sqrt{\rho_k(x_k)}.
    \label{eq:optical-two-momenta-app}
\end{equation}

This makes explicit the distinction between a branchwise
Hamilton--Jacobi description and the Morse-family description. If one
eliminates the ray parameter, the two branches must be represented
separately by local generating functions
\begin{equation}
    W_k^+(x_k),
    \qquad
    W_k^-(x_k),
    \label{eq:optical-separate-W-app}
\end{equation}
satisfying
\begin{equation}
    \frac{dW_k^\pm}{dx_k}
    =
    p_k^\pm(x_k).
    \label{eq:optical-Wpm-app}
\end{equation}

More generally, if the projection produces several sheets, a branchwise
representation requires a collection
\[
    W_k^1,\ldots,W_k^N,
\]
with
\[
    \Lambda_k^d
    =
    \bigcup_{j=1}^{N}
    \operatorname{Im}(dW_k^j)
\]
on the corresponding regular regions.

The Morse-family representation does not introduce these functions
independently. Instead, one keeps the single function
\[
    S_k(x_k,a)
\]
and its critical set
\[
    \Sigma_k
    =
    \left\{
        (x_k,a)
        \; ;\;
        \frac{\partial S_k}{\partial a}=0
    \right\}.
\]
Each critical point \(a_j(x_k)\) determines one sheet through
\begin{equation}
    p_k^j(x_k)
    =
    \frac{\partial S_k}{\partial x_k}
    \bigl(x_k,a_j(x_k)\bigr).
    \label{eq:optical-critical-sheet-app}
\end{equation}
Thus the different branches appear as different critical points of the
same generating family rather than as independently defined generating
functions.

For the fold considered here,
\[
    \Sigma_k
    =
    \left\{
        (x_k,a)
        \; ;\;
        (a-a_c)^2=\rho_k(x_k)
    \right\},
\]
and its two regular components give
\[
    a_+(x_k),
    \qquad
    a_-(x_k).
\]
At the caustic,
\[
    \rho_k(x_c)=0,
\]
these critical points coalesce:
\[
    a_+(x_c)=a_-(x_c)=a_c.
\]
The separate functions \(W_k^\pm\) cease to provide a single smooth
description of the complete Lagrangian submanifold, whereas
\(S_k(x_k,a)\) remains smooth.

\subsection{Propagation of the complete Morse family}

We now propagate this Lagrangian submanifold without first decomposing it
into its branches.

Given \(S_k(x_k,a)\), define
\begin{equation}
\begin{aligned}
    \mathcal S_{k+1}
    (x_{k+1};x_k,p_{k+1},a)
    ={}&
    S_k(x_k,a)
    +
    p_{k+1}x_{k+1}
    -
    H_{d+}^{\,k}(x_k,p_{k+1}).
    \label{eq:optical-composition-family-app}
\end{aligned}
\end{equation}
The variables \(x_k,p_{k+1}\), and \(a\) are auxiliary variables of this
composition, while \(x_{k+1}\) is the base variable of the propagated
Lagrangian submanifold.

The criticality equations are
\begin{equation}
    \frac{\partial\mathcal S_{k+1}}{\partial x_k}=0,
    \qquad
    \frac{\partial\mathcal S_{k+1}}{\partial p_{k+1}}=0,
    \qquad
    \frac{\partial\mathcal S_{k+1}}{\partial a}=0.
    \label{eq:optical-composition-critical-app}
\end{equation}

The first equation gives
\begin{equation}
    \frac{\partial S_k}{\partial x_k}
    =
    \frac{\partial H_{d+}^{\,k}}{\partial x_k}.
    \label{eq:optical-composition-xk-app}
\end{equation}
Since
\[
    p_k=\frac{\partial S_k}{\partial x_k},
\]
we recover
\[
    p_k
    =
    \frac{\partial H_{d+}^{\,k}}{\partial x_k}.
\]

The second equation gives
\begin{equation}
    x_{k+1}
    =
    \frac{\partial H_{d+}^{\,k}}{\partial p_{k+1}},
    \label{eq:optical-composition-p-app}
\end{equation}
which is the second Type--II discrete Hamilton equation.

Finally,
\begin{equation}
    \frac{\partial\mathcal S_{k+1}}{\partial a}
    =
    \frac{\partial S_k}{\partial a}
    =
    0,
    \label{eq:optical-composition-a-app}
\end{equation}
so the complete critical set of the incoming Morse family is retained
during the propagation.

The momentum of the propagated Lagrangian submanifold is
\begin{equation}
    p_{k+1}
    =
    \frac{\partial\mathcal S_{k+1}}
         {\partial x_{k+1}}.
    \label{eq:optical-propagated-momentum-app}
\end{equation}

Consequently,
\begin{equation}
    \Lambda_{k+1}^d
    =
    \left\{
        \left(
            x_{k+1},
            \frac{\partial\mathcal S_{k+1}}
                 {\partial x_{k+1}}
        \right)
        \; ;\;
        D_{(x_k,p_{k+1},a)}
        \mathcal S_{k+1}=0
    \right\}.
    \label{eq:optical-propagated-Lagrangian-app}
\end{equation}

For the optical Hamiltonian, the composition family becomes
\begin{equation}
\begin{aligned}
    \mathcal S_{k+1}
    ={}&
    S_k(x_k,a)
    +
    p_{k+1}(x_{k+1}-x_k)
    \\
    &+
    h\sqrt{
        n^2(z_k,x_k)-p_{k+1}^2
    }.
    \label{eq:optical-composition-explicit-app}
\end{aligned}
\end{equation}
Its criticality equations are
\begin{equation}
    (a-a_c)^2-\rho_k(x_k)=0,
    \label{eq:optical-algorithm-a-app}
\end{equation}
\begin{equation}
    S_{\mathrm{reg},k}'(x_k)
    -
    \rho_k'(x_k)(a-a_c)
    -
    p_{k+1}
    +
    h
    \frac{
        n(z_k,x_k)n_x(z_k,x_k)
    }{
        \sqrt{n^2(z_k,x_k)-p_{k+1}^2}
    }
    =
    0,
    \label{eq:optical-algorithm-p-app}
\end{equation}
and
\begin{equation}
    x_{k+1}
    -
    x_k
    -
    h
    \frac{
        p_{k+1}
    }{
        \sqrt{n^2(z_k,x_k)-p_{k+1}^2}
    }
    =
    0.
    \label{eq:optical-algorithm-x-app}
\end{equation}

Using
\[
    p_k
    =
    S_{\mathrm{reg},k}'(x_k)
    -
    \rho_k'(x_k)(a-a_c),
\]
the last two equations are exactly
\eqref{eq:optical-discrete-p-app} and
\eqref{eq:optical-discrete-x-app}. Hence the Morse-family propagation
introduces no modification of the underlying Type--II integrator; it
provides a representation of the complete Lagrangian submanifold on which
that integrator acts.

In particular, the two solutions
\[
    a=a_\pm(x_k)
\]
of \eqref{eq:optical-algorithm-a-app} are propagated by the same critical
system \eqref{eq:optical-algorithm-a-app}--
\eqref{eq:optical-algorithm-x-app}. There is no need to construct two
different propagation laws for \(W_k^+\) and \(W_k^-\).

This is the central computational distinction:
\[
    \{W_k^+,W_k^-,\ldots\}
    \qquad\hbox{is replaced geometrically by}\qquad
    \left(
        S_k(x_k,a),
        \frac{\partial S_k}{\partial a}=0
    \right).
\]
The different sheets are recovered as the different critical points of the
single Morse family.

\subsection{Verification of the regularity and reduction hypotheses}
\label{sec:optical-hypotheses}

Here $\dim\Q=1$ --- we avoid writing $n$ for this, since $n$ denotes the refractive index
throughout the present section --- the incoming family $S_k$ of
\eqref{eq:optical-Sk-app} carries the single auxiliary variable $a$, so $m_a=1$, and the
Type--II discrete Hamiltonian \eqref{eq:optical-Hdplus-explicit-app} carries none, so
$m_\lambda=0$. We check in
turn the hypotheses of Propositions~\ref{prop:composition} and
\ref{prop:reduction}, all of them on a compact subset of the domain $\mathcal D$ of
\eqref{eq:optical-domain}.

First, $S_k$ is a Morse family. On its critical set $(a-a_c)^2=\rho_k(x_k)$,
\[
\left(
\frac{\partial^2S_k}{\partial a\,\partial x_k}
\quad
\frac{\partial^2S_k}{\partial a\,\partial a}
\right)
=
\bigl(-\rho_k'(x_k)\quad 2(a-a_c)\bigr),
\]
which has rank one unless $\rho_k'(x_k)=0$ and $a=a_c$ simultaneously. The second
equality forces $\rho_k(x_k)=0$, hence $x_k=x_c$, and the first is then excluded by
the fold hypothesis \eqref{eq:optical-rho-app}. Thus $S_k$ is a Morse family
everywhere, including at the caustic itself.

Second, $F_{qp}$ is invertible for small $h$. From
\eqref{eq:optical-Hdplus-explicit-app},
\[
F_{qp}
=
\frac{\partial^2H_{d+}^{\,k}}{\partial x_k\,\partial p_{k+1}}
=
1
+
h\,\frac{\partial^2H}{\partial x\,\partial p}
=
1
-
h\,
\frac{p_{k+1}\,n(z_k,x_k)n_x(z_k,x_k)}
     {\bigl(n^2(z_k,x_k)-p_{k+1}^2\bigr)^{3/2}} ,
\]
which is non-zero as soon as $h$ is smaller than the reciprocal of the supremum of
$|p\,nn_x|(n^2-p^2)^{-3/2}$ over the compact set in question. Since $m_\lambda=0$ the
second condition in \eqref{eq:composition-hypotheses} is vacuous, so
Proposition~\ref{prop:composition} applies and the composition family
\eqref{eq:optical-composition-family-app} is a Morse family generating the propagated
Lagrangian submanifold \eqref{eq:optical-propagated-Lagrangian-app}.

Third, the reduction applies. The discrete Hamiltonian is of the consistent form
\eqref{eq:consistent-family} with $G_k=H(z_k,\cdot,\cdot)$, whose second derivatives
are bounded on compact subsets of $\mathcal D$, and
\[
\frac{\partial^2S_k}{\partial x_k^2}
=
S_{\mathrm{reg},k}''(x_k)-\rho_k''(x_k)(a-a_c)
\]
is bounded once $a$ is also confined to a compact neighbourhood of the critical set ---
a compact set in $(z,x,p)$ alone does not bound it. On such a set both $C_G$ and
$C_S$ of Proposition~\ref{prop:reduction} are finite and the reduction is available
for $h<h_0$. The estimate is for a single step; to iterate it we assume that the
propagated family remains in a compact chart on which the same domain, smoothness,
twist and Hessian bounds hold, and under that assumption the auxiliary variables $x_k$
and $p_{k+1}$ introduced by the composition can be eliminated at each step, so that the
propagation is a one-step map on generating families of a single auxiliary variable,
\[
S_k(x_k,a)
\longmapsto
\widetilde S_{k+1}(x_{k+1},a),
\]
rather than a construction whose auxiliary dimension grows with $k$. Explicitly, by
Corollary~\ref{cor:semidiscrete-HJ},
\begin{equation}
\widetilde S_{k+1}(x,a)
=
S_k(x,a)
+
h\sqrt{
    n^2(z_k,x)
    -
    \left(\frac{\partial S_k}{\partial x}(x,a)\right)^{\!2}
}
+
\mathcal O(h^2),
\label{eq:optical-reduced-update}
\end{equation}
the plus sign arising because the update subtracts $hH$ and the optical Hamiltonian
\eqref{eq:optical-Hamiltonian-app} is itself negative. Equation \eqref{eq:optical-reduced-update} is a
first-order discretisation of the eikonal equation for the family $S(x,a)$, and by
Remark~\ref{rem:through-caustics} it remains valid through the fold, where the
corresponding equation for a single-valued $W$ does not.

\subsection{Recovery of the discrete Hamilton--Jacobi equations on regular branches}

Away from the caustic, the projection of the Lagrangian submanifold onto
configuration space is locally regular. Hence, each local branch may be
represented by an exact one-form,
\[
p_k=\frac{dW_k^j}{dx_k},
\qquad
p_{k+1}=\frac{dW_{k+1}^j}{dx_{k+1}},
\]
where the index \(j\) labels a local graphical branch.

For the Type--II discrete Hamiltonian
\[
H_{d+}^{\,k}(x_k,p_{k+1}),
\]
which carries no auxiliary variable, the critical set $\mathcal C_k$ of
\eqref{eq:HJ-theorem-1} is parametrised by $(x_k,p_{k+1})$ alone and has a single branch
over each regular point, so the branchwise form of the discrete Hamilton--Jacobi
equations of Section~\ref{sec:discrete-HJ} is here equivalent to
\eqref{eq:exact-HJ-1}, and reads
\begin{equation}
\frac{dW_k^j}{dx_k}
=
\frac{\partial H_{d+}^{\,k}}{\partial x_k}
\left(
x_k,
\frac{dW_{k+1}^j}{dx_{k+1}}
\right),
\label{eq:optical-HJ-1}
\end{equation}
together with
\begin{equation}
x_{k+1}
=
\frac{\partial H_{d+}^{\,k}}{\partial p_{k+1}}
\left(
x_k,
\frac{dW_{k+1}^j}{dx_{k+1}}
\right).
\label{eq:optical-HJ-2}
\end{equation}

Thus, on every regular branch, the Morse-family propagation reduces locally
to the ordinary Type--II discrete Hamilton--Jacobi description relating the
generating functions at two consecutive steps,
\[
W_k^j
\qquad\text{and}\qquad
W_{k+1}^j.
\]

At the caustic, however, the projection ceases to be locally invertible.
Consequently, the complete propagated Lagrangian submanifold can no longer
be represented by a single generating function \(W_k\). Although one could
introduce separate local functions
\[
W_k^1,\ldots,W_k^N
\]
on the individual graphical branches, such a description requires an
explicit branch decomposition.

The Morse-family formulation avoids this decomposition. Instead, the entire
Lagrangian submanifold is represented by the single generating family
\[
S_k(x_k,a),
\]
whose different branches arise automatically as different critical points
of
\[
\frac{\partial S_k}{\partial a}=0.
\]
The propagation is then performed through the composition family
\begin{equation}
\mathcal S_{k+1}
=
S_k(x_k,a)
+
p_{k+1}x_{k+1}
-
H_{d+}^{\,k}(x_k,p_{k+1}),
\label{eq:optical-composition-final}
\end{equation}
with criticality conditions
\begin{equation}
\frac{\partial \mathcal S_{k+1}}{\partial x_k}=0,
\qquad
\frac{\partial \mathcal S_{k+1}}{\partial p_{k+1}}=0,
\qquad
\frac{\partial \mathcal S_{k+1}}{\partial a}=0.
\label{eq:optical-composition-critical-final}
\end{equation}

Therefore, the proposed method does not propagate the branches
\(W_k^1,W_k^2,\ldots\) separately. Rather, it propagates a single Morse
family whose critical set contains all the branches simultaneously.
The Type--II discrete Hamiltonian provides the symplectic propagation in
phase space, whereas the Morse family provides a regular representation of
the multivalued Hamilton--Jacobi solution when its projection onto
configuration space develops a caustic.

Accordingly, the construction may be interpreted as a Morse-family
integrator for discrete Hamilton--Jacobi dynamics across caustics: the
discrete dynamics propagates the complete Lagrangian submanifold, while the
different branches of the Hamilton--Jacobi solution are recovered as
different critical points of the same generating family.

\bibliographystyle{plainnat}
\bibliography{references} 

\end{document}